%% file: bearing_only.tex
\documentclass[journal]{IEEEtran}

\include{before_document}

\begin{document}

\title{Bearing Rigidity and Almost Global\\ Bearing-Only Formation Stabilization}

\author{Shiyu Zhao and Daniel Zelazo
\thanks{S. Zhao and D. Zelazo are with the Faculty of Aerospace Engineering, Israel Institute of Technology, Haifa, Israel.
    {\tt\small szhao@tx.technion.ac.il, dzelazo@technion.ac.il}}
}%\author

\maketitle

\begin{abstract}
A fundamental problem that the bearing rigidity theory studies is to determine when a framework can be uniquely determined up to a translation and a scaling factor by its inter-neighbor bearings.
While many previous works focused on the bearing rigidity of two-dimensional frameworks, a first contribution of this paper is to extend these results to arbitrary dimensions.
It is shown that a framework in an arbitrary dimension can be uniquely determined up to a translation and a scaling factor by the bearings if and only if the framework is infinitesimally bearing rigid.
In this paper, the proposed bearing rigidity theory is further applied to the bearing-only formation stabilization problem where the target formation is defined by inter-neighbor bearings and the feedback control uses only bearing measurements.
Nonlinear distributed bearing-only formation control laws are proposed for the cases with and without a global orientation.
It is proved that the control laws can almost globally stabilize infinitesimally bearing rigid formations.
Numerical simulations are provided to support the analysis.
\end{abstract}


\IEEEpeerreviewmaketitle

\section{Introduction}

%\subsection{Background}

Multi-agent formation control has been studied extensively in recent years with
\emph{distance-constrained formation control} taking a prominent role \cite{Anderson2008CSM,Krick2009IJC,Corte2009Automatica,Dorfler2009ECC,Dimos2010Automatica,oh2013IJRNC,SunZhiyong2014IFAC}.  In this setting it is assumed that the target formation is specified by inter-agent distances, and each agent is able to measure relative positions of their neighbors.
\emph{Bearing-constrained formation control} has also attracted much attention recently \cite{bishopconf2011rigid,bishop2010SCL,Eren2012IJC,Antonio2012CDC,zhao2013SCLDistribued,zhao2013IJCFinite,Eric2014ACC}.   Instead of distances, the formation is specified by inter-agent bearings, and each agent can measure the relative positions or bearings of their neighbors.
Bearing measurements are often cheaper and more accessible than position measurements, spurring interest in cooperative control using bearing-only measurements \cite{nima2009TR,bishop2010SCL,Eren2012IJC,Antonio2012CDC,Franchi2012IJRR,ZhengRonghao2013SCL,zhao2013SCLDistribued,zhao2013IJCFinite,Cornejo2013IJRR,Eric2014ACC}.
%Bearing-based formation control can be potentially applied to vision-based cooperative control of multi-vehicle systems where each vehicle can measure the bearings of their neighbors with a camera.

This paper studies a bearing-only formation control problem where the target formation is \emph{bearing-constrained} and each agent has access to the \emph{bearing-only} measurements of their neighbors.
It is noted that while bearing measurements can be used to estimate relative distances or positions \cite{Franchi2012IJRR,Cornejo2013IJRR,zelazo2014SE2Rigidity}, such schemes may significantly increase the complexity of the sensing system in terms of both hardware and software.  This then motivates our study focusing on a pure bearing-only control scheme without the need for estimation of additional quantities (e.g., relative position).%, where the bearing measurements are directly applied in the formation control and it is not required to estimate additional quantities (e.g., relative position).

%\subsection{Literature Review}

Although bearing-only formation control has lately attracted much interest, many problems on this topic remain unsolved.
The studies in \cite{nima2009TR,bishopconf2011rigid,Antonio2012CDC} considered bearing-constrained formation control in two-dimensional spaces, but required access to position or other measurements in the proposed control laws.
The results reported in \cite{Franchi2012IJRR,Cornejo2013IJRR} only require bearing measurements, but they are used to estimate additional relative-state information such as distance ratios or scale-free coordinates.
The works in \cite{bishop2010SCL,Eren2012IJC,zhao2013SCLDistribued,zhao2013IJCFinite} studied formation control with bearing measurements directly applied in the control.  However, these results were applied to special formations, such as cyclic formations, and may not be extendable to arbitrary formation shapes.
A very recent work reported in \cite{Eric2014ACC} solved bearing-only formation control for arbitrary underlying graphs, but only for formations in the plane.
Bearing-only formation control in arbitrary dimensions with general underlying graphs still remains an open problem.

%\subsection{Main Contributions}
A central tool in the study of bearing-only formation control is \emph{bearing rigidity theory}\footnote{Also referred to as \emph{parallel rigidity} in some literature.}.
Existing works on bearing rigidity mainly focused on frameworks in two-dimensional ambient spaces \cite{eren2003,bishopconf2011rigid,Eren2012IJC,zelazo2014SE2Rigidity}.  The first contribution of our work, therefore, is an extension of the existing bearing rigidity theory to arbitrary dimensions.  We also explore connections between bearing rigidity and distance rigidity, and in particular show that a framework in $\R^2$ is infinitesimally bearing rigid if and only if it is also infinitesimally distance rigid.

Based on the proposed bearing rigidity theory, we investigate distributed bearing-only formation control in arbitrary dimensions in the presence of a global reference frame.
We propose a distributed bearing-only formation control law and show by a Lyapunov approach that the control law can almost globally stabilize infinitesimally bearing rigid formations.  We also provide a sufficient condition ensuring collision avoidance between any pair of agents under the action of the control.

In the third part of the paper, we investigate bearing-only formation control in the three dimensional space without a global reference frame known to the agents.
Each agent can only measure the bearings and relative orientations of their neighbors in their local reference frames.
We propose a distributed control law to control both the position and the orientation of each agent.
It is shown that the orientation will synchronize and the target formation is almost globally stable.

%\subsection{Organization}

This paper is organized as follows.
Section~\ref{section_parallelRigidity} presents the bearing rigidity theory that is applicable to arbitrary dimensions.
Section~\ref{section_BOF_global} studies bearing-only formation control in arbitrary dimensions in the presence of a global reference frame, and Section~\ref{section_BOF_noGlobal} studies the case without a global reference frame.
Simulation results are presented in Section~\ref{section_simulation}.
Conclusions and future works are given in Section~\ref{section_conclusion}.

\paragraph*{Notations}

Given $A_i\in\mathbb{R}^{p\times q}$ for $i=1,\dots,n$, denote $ \mydiag(A_i)\triangleq\blkdiag\{A_1,\dots,A_n\}\in\mathbb{R}^{np\times nq}$.
Let $\Null(\cdot)$, $\Range(\cdot)$, and $\rank(\cdot)$ be the null space, range space, and rank of a matrix, respectively.
Denote $I_d\in\R^{d\times d}$ as the identity matrix, and $\one\triangleq[1,\dots,1]^\T$.
%For a positive semi-definite matrix $A\in\R^{d\times d}$, its eigenvalues are denoted as $0\le\lambda_1(A)\le\lambda_2(A)\le\dots\le\lambda_d(A)$.
Let $\|\cdot\|$ be the Euclidian norm of a vector or the spectral norm of a matrix, and $\otimes$ the Kronecker product.
For any $x=[x_1,x_2,x_3]^\T\in\R^3$, the associated skew-symmetric matrix is denoted as
\begin{align}\label{eq_skewSymmetricOperator}
    \sk{x}\triangleq\left[
      \begin{array}{ccc}
        0 & -x_3 & x_2 \\
        x_3 & 0 & -x_1 \\
        -x_2 & x_1 & 0 \\
      \end{array}
    \right].
\end{align}

An undirected graph, denoted as $\mathcal{G}=(\mathcal{V},\mathcal{E})$, consists of a vertex set $\mathcal{V}=\{1,\dots,n\}$ and an edge set $\mathcal{E}\subseteq \mathcal{V} \times \mathcal{V}$ with $m=|\mathcal{E}|$.
The set of neighbors of vertex $i$ is denoted as $\mathcal{N}_i\triangleq\{j \in \mathcal{V}: \ (i,j)\in \mathcal{E}\}$.
An {orientation} of an undirected graph is the assignment of a direction to each edge.
An {oriented graph} is an undirected graph together with an orientation.
The {incidence matrix} $H\in\mathbb{R}^{m\times n}$ of an oriented graph is the $\{0,\pm1\}$-matrix with rows indexed by
edges and columns by vertices: $[H]_{ki}=1$ if vertex $i$ is the head of edge $k$, $[H]_{ki}=-1$ if vertex $i$ is the tail of edge $k$, and $[H]_{ki}=0$ otherwise.  For a connected graph, one always has $H\one = 0$ and $\rank(H)=n-1$ \cite{Mesbahi2010}.

\section{Bearing Rigidity in Arbitrary Dimensions}\label{section_parallelRigidity}

In this section, we propose a bearing rigidity theory that is applicable to arbitrary dimensions.
The basic problem that the bearing rigidity theory studies is whether a framework can be uniquely determined up to a translation and a scaling factor given the bearings between each pair of neighbors in the framework.
This problem can be equivalently stated as whether two frameworks with the same inter-neighbor bearings have the same shape.

We first define some necessary notations.
Given a finite collection of $n$ points $\{p_i\}_{i=1}^n$ in $\R^d$ ($n\ge2$, $d\ge2$), a \emph{configuration} is denoted as $p=[p_1^\T,\dots,p_n^\T]^\T\in\mathbb{R}^{dn}$.
A \emph{framework} in $\R^d$, denoted as $\G(p)$, is a combination of an undirected graph $\G=(\V,\E)$ and a configuration $p$, where vertex $i\in\V$ in the graph is mapped to the point $p_i$ in the configuration.
For a framework $\G(p)$, define
\begin{align}\label{eq_eij_gij_def}
    e_{ij}\triangleq p_j-p_i, \quad g_{ij}\triangleq e_{ij}/\|e_{ij}\|, \quad \forall (i,j)\in\E.
\end{align}
Note the unit vector $g_{ij}$ represents the relative bearing of $p_j$ to $p_i$.
This unit-vector representation is different from the conventional ways where a bearing is described as one angle (azimuth) in $\mathbb{R}^2$, or two angles (azimuth and altitude) in $\mathbb{R}^3$.
Note also that $e_{ij}=-e_{ji}$ and $g_{ij}=-g_{ji}$.

We now introduce an important orthogonal projection operator that will be widely used in this paper.
For any nonzero vector $x\in\R^d$ ($d\ge2$), define the operator $P: \R^d\rightarrow\R^{d\times d}$ as
\begin{align*}
    P(x) \triangleq I_d - \frac{x}{\|x\|}\frac{x^\T }{\|x\|}.
\end{align*}
For notational simplicity, denote $P_x=P(x)$.
Note $P_x$ is an orthogonal projection matrix which geometrically projects any vector onto the orthogonal compliment of $x$.
It can be verified that $P_x^\T =P_x$, $P_x^2=P_x$, and $P_x$ is positive semi-definite.
Moreover, $\Null(P_x)=\myspan\{x\}$ and the eigenvalues of $P_x$ are $\{0,1^{(d-1)}\}$.
In the bearing rigidity theory, it is often required to evaluate whether two given bearings are parallel to each other.
The orthogonal projection operator provides a convenient way to describe parallel vectors in arbitrary dimensions.
\begin{lemma}\label{lemma_parallelismDefinition}
    Two nonzero vectors $x,y\in\R^d$ are parallel if and only if $P_x y=0$ (or equivalently $P_{y} x=0$).
\end{lemma}
\begin{proof}
The result follows from $\Null(P_x)=\myspan\{x\}$.
\end{proof}

\begin{remark}
Most existing works use the notion of \emph{normal vectors} to describe parallel vectors in $\R^2$ \cite{eren2003,bishopconf2011rigid,Eren2012IJC}.
Specifically, given a nonzero vector $x\in\R^2$, denote $x^\perp\in\R^2$ as a nonzero normal vector satisfying $x^\T x^\perp =0$. Then any vector $y\in\R^2$ is parallel to $x$ if and only if $(x^\perp)^\T y=0$.
This approach is applicable to two dimensional cases but difficult to extend to arbitrary dimensions.
%The orthogonal projection operator provides a general way to characterize parallel vectors in arbitrary dimensions.
Moreover, it is straightforward to prove that in $\R^2$ the use of the orthogonal projection operator is equivalent to the use of normal vectors based on the fact that $P_x=x^\perp(x^\perp)^\T/\|x^\perp\|^2$ for $x\in\R^2$.
\end{remark}

We are now ready to define the fundamental concepts in bearing rigidity.
These concepts are defined analogously to those in the distance rigidity theory.

\begin{definition}[Bearing Equivalency]\label{def.bearingequiv}
    Frameworks $\G(p)$ and $\G(p')$ are \emph{bearing equivalent} if $P_{(p_i-p_j)}(p_i'-p_j')=0$ for all $(i,j)\in\E$.
\end{definition}

\begin{definition}[Bearing Congruency]\label{def.bearingcong}
    Frameworks $\G(p)$ and $\G(p')$ are \emph{bearing congruent} if $P_{(p_i-p_j)}(p_i'-p_j')=0$ for all $i,j\in\V$.
\end{definition}

\begin{figure}[t]
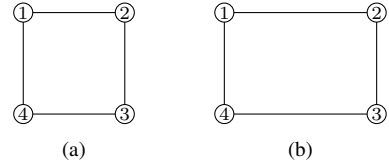

  \centering
  \def\myscale{0.45}
  \include{figure_tikz_equivNotCongru}
  \caption{The two frameworks are bearing equivalent but \emph{not} bearing congruent.
   The bearings between $(p_1,p_3)$ or $(p_2,p_4)$ of the frameworks are different.}
   \vspace{-15pt}
  \label{fig_equivNotCongru}
\end{figure}

By definition, bearing congruency implies bearing equivalency.
The converse, however, is not true, as illustrated in Figure~\ref{fig_equivNotCongru}.

\begin{definition}[Bearing Rigidity]
    A framework $\G(p)$ is \emph{bearing rigid} if there exists a constant $\epsilon>0$ such that any framework $\G(p')$ that is bearing equivalent to $\G(p)$ and satisfies $\|p'-p\|<\epsilon$ is also bearing congruent to $\G(p)$.
\end{definition}

\begin{definition}[Global Bearing Rigidity]\label{definition_globalParallelRigidity}
    A framework $\G(p)$ is \emph{globally bearing rigid} if an arbitrary framework that is bearing equivalent to $\G(p)$ is also bearing congruent to $\G(p)$.
\end{definition}

By definition, global bearing rigidity implies bearing rigidity.
As will be shown later, the converse is also true.

We next define \emph{infinitesimal bearing rigidity}, which is one of the most important concepts in the bearing rigidity theory.
%Recall the graph $\G$ is assumed to be undirected.
Consider an arbitrary orientation of the graph $\G$ and denote
\begin{align}\label{eq_ek_gk_def}
    e_{k}\triangleq p_j-p_i,\quad g_{k}\triangleq {e_{k}}/{\|e_{k}\|}, \quad \forall k\in\{1,\dots,m\}
\end{align}
as the edge vector and the bearing for the $k$th directed edge.  % of the oriented graph.
  Denote $e=[e_1^\T ,\dots,e_m^\T ]^\T$ and $g=[g_1^\T ,\dots,g_m^\T ]^\T$.
Note $e$ satisfies $e=\bar{H}p$ where $\bar{H}=H\otimes I_d$. % and $H$ is the incidence matrix.
Define the \emph{bearing function} $F_B: \R^{dn}\rightarrow\R^{dm}$ as
\begin{align*}%\label{eq_bearingEdgeFunctionDefinition1}
    F_B(p)\triangleq
    \left[
            \begin{array}{ccc}
              g_1^\T & \cdots & g_m^\T
            \end{array}
    \right]^\T\in\R^{dm}.
\end{align*}
The bearing function describes all the bearings in the framework.
The \emph{bearing rigidity matrix} is defined as the Jacobian of the bearing function,
\begin{align}\label{eq_rigidityMatrixDefinition}
    R(p) \triangleq \frac{\partial F_B(p)}{\partial p}\in\R^{dm\times dn}.
\end{align}
Let $\delta p$ be a variation of the configuration $p$.
If $R(p)\delta p=0$, then $\delta{p}$ is called an \emph{infinitesimal bearing motion} of $\G(p)$.  This is analogous to infinitesimal motions in distance-based rigidity.
Distance preserving motions of a framework include rigid-body translations and rotations, whereas bearing preserving motions of a framework include translations and scalings.
An infinitesimal bearing motion is called \emph{trivial} if it corresponds to a translation and a scaling of the entire framework.

\begin{definition}[Infinitesimal Bearing Rigidity]\label{definition_infinitesimalParallelRigid}
    A framework is \emph{infinitesimally bearing rigid} if all the infinitesimal bearing motions are trivial.
\end{definition}

Up to this point, we have introduced all the fundamental concepts in the bearing rigidity theory.
We next explore the properties of these concepts.
We first derive a useful expression for the bearing rigidity matrix.

\begin{lemma}
The bearing rigidity matrix in \eqref{eq_rigidityMatrixDefinition} can be expressed as
\begin{align}\label{eq_rigidityMatrixForm}
    R(p) = \mydiag\left(\frac{P_{g_k}}{\|e_k\|}\right)\bar{H}.
\end{align}
\end{lemma}

\begin{proof}
%Recall the bearing function is $F_B(p)=[g_1^\T,\dots,g_m^\T]^\T$.
    It follows from $g_k=e_k/\|e_k\|, \forall k\in\{1,\dots,m\}$ that
    \begin{align*}
        \frac{\partial g_k}{\partial e_k} = \frac{1}{\|e_k\|}\left(I_d - \frac{e_k}{\|e_k\|}\frac{e_k^\T }{\|e_k\|}\right)= \frac{1}{\|e_k\|}P_{g_k}.
    \end{align*}
    As a result, $\partial F_B(p)/\partial e=\dia{P_{g_k}/\|e_k\|}$ and consequently
    \begin{align*}
        R(p)=\frac{\partial F_B(p)}{\partial p}
        =\frac{\partial F_B(p)}{\partial e}\frac{\partial e}{\partial p}
        =\mydiag\left(\frac{P_{g_k}}{\|e_k\|}\right)\bar{H}.
    \end{align*}
\end{proof}
The expression \eqref{eq_rigidityMatrixForm} can be used to characterize the null space and the rank of the bearing rigidity matrix.

\begin{lemma}\label{lemma_nullSpaceRigidity}
    A framework $\G(p)$ in $\R^d$ always satisfies $\myspan\{\one\otimes I_d, p\}\subseteq \Null(R(p))$ and $\rank(R(p))\le dn-d-1$.
\end{lemma}

\begin{proof}
    First, it is clear that $\myspan\{\one\otimes I_d\}\subseteq\Null(\bar{H})\subseteq\Null(R(p))$.
    Second, since $P_{e_k}e_k=0$, we have $R(p)p=\mydiag(P_{e_k}/\|e_k\|)\bar{H}p=\mydiag(P_{e_k}/\|e_k\|)e=0$ and hence $p\subseteq\Null(R(p))$.
    The inequality $\rank(R(p))\le dn-d-1$ follows immediately from $\myspan\{\one\otimes I_d, p\}\subseteq \Null(R(p))$.
\end{proof}

For any undirected graph $\G=(\V,\E)$, denote $\G^\kappa$ as the complete graph over the same vertex set $\V$, and $R^\kappa(p)$ as the bearing rigidity matrix of the framework $G^\kappa(p)$.
The next result gives the necessary and sufficient conditions for bearing equivalency and bearing congruency.

\begin{theorem}\label{theorem_NSConditionForBEBC}
    Two frameworks $\G(p)$ and $\G(p')$ are bearing equivalent if and only if $R(p)p'=0$, and bearing congruent if and only if $R^\kappa(p)p'=0$.
\end{theorem}

\begin{proof}
    Since $R(p)p'=\dia{I_d/\|e_k\|}\dia{P_{g_k}}\bar{H}p'=\dia{I_d/\|e_k\|}\dia{P_{g_k}}e'$, we have
    \begin{align*}
        R(p)p'=0 \,\Leftrightarrow P_{g_k}e_k'=0, \,\forall k\in\{1,\dots,m\}.
    \end{align*}
    Therefore, by Definition \ref{def.bearingequiv}, the two frameworks are bearing equivalent if and only if $R(p)p'=0$.
    By Definition \ref{def.bearingcong}, it can be analogously shown that frameworks are bearing equivalent if and only if $R^\kappa(p)p'=0$.
\end{proof}

Since any infinitesimal motion $\delta p$ is in $\Null(R(p))$, it is implied from Theorem~\ref{theorem_NSConditionForBEBC} that $R(p)(p+\delta p)=0$ and hence $\G(p+\delta p)$ is bearing equivalent to $\G(p)$.

We next give a useful lemma and then prove the necessary and sufficient condition for global bearing rigidity.

\begin{lemma}\label{lemma_NullRkappa_NullR}
A framework $\G(p)$ in $\R^d$ always satisfies $\myspan\{\one\otimes I_d, p\}\subseteq\Null(R^\kappa(p))\subseteq\Null(R(p))$ and $dn-d-1\ge\rank(R^\kappa(p))\ge\rank(R(p))$.
\end{lemma}
\begin{proof}
The results that $\myspan\{\one\otimes I_d, p\}\subseteq\Null(R^\kappa(p))$ and $dn-d-1\ge\rank(R^\kappa(p))$ can be proved similarly as Lemma~\ref{lemma_nullSpaceRigidity}.
    For any $\delta p\in\Null(R^\kappa(p))$, we have $R^\kappa(p)\delta p=0\Rightarrow R^\kappa(p)(p+\delta p)=0$.
    As a result, $\G(p+\delta p)$ is bearing congruent to $\G(p)$ by Theorem~\ref{theorem_NSConditionForBEBC}.
    Since bearing congruency implies bearing equivalency, we further know $R(p)(p+\delta p)=0$ and hence $R(p)\delta p=0$.
    Therefore, any $\delta p$ in $\Null(R^\kappa(p))$ is also in $\Null(R(p))$ and thus $\Null(R^\kappa(p))\subseteq\Null(R(p))$.
    Since $R(p)$ and $R^\kappa(p)$ have the same column number, it follows immediately that $\rank(R^\kappa(p))\ge\rank(R(p))$.
\end{proof}

\begin{theorem}[Condition for Global Bearing Rigidity]\label{theorem_NSConditionforGBR}
    A framework $\G(p)$ in $\R^d$ is globally bearing rigid if and only if $\Null(R^\kappa(p))=\Null(R(p))$ or equivalently $\rank(R^\kappa(p))=\rank(R(p))$.
\end{theorem}
\begin{proof}
    \emph{(Necessity)}
    Suppose the framework $\G(p)$ is globally bearing rigid.
    We next show that $\Null(R(p))\subseteq\Null(R^\kappa(p))$.
    For any $\delta p\in \Null(R(p))$, we have $R(p)\delta p=0\Rightarrow R(p)(p+\delta p)=0$.
    As a result, $\G(p+\delta p)$ is bearing equivalent to $\G(p)$ according to Theorem~\ref{theorem_NSConditionForBEBC}.
    Since $\G(p)$ is globally bearing rigid, it follows that $\G(p+\delta p)$ is also bearing congruent to $\G(p)$, which means $R^\kappa(p)(p+\delta p)=0\Rightarrow R^\kappa(p)\delta p=0$.
    Therefore, any $\delta p$ in $\Null(R(p))$ is in $\Null(R^\kappa(p))$ and thus $\Null(R(p))\subseteq\Null(R^\kappa(p))$.
    Since $\Null(R^\kappa(p))\subseteq\Null(R(p))$ as shown in Lemma~\ref{lemma_NullRkappa_NullR}, we have $\Null(R(p))=\Null(R^\kappa(p))$.

    \emph{(Sufficiency)} Suppose $\Null(R(p))=\Null(R^\kappa(p))$.
    Any framework $\G(p')$ that is bearing equivalent to $\G(p)$ satisfies $R(p)p'=0$.
    It then follows from $\Null(R(p))=\Null(R^\kappa(p))$ that $R^\kappa(p)p'=0$, which means $\G(p')$ is also bearing congruent to $\G(p)$.
    As a result, $\G(p)$ is globally bearing rigid.

    Because $R(p)$ and $R^\kappa(p)$ have the same column number, it follows immediately that $\Null(R^\kappa(p))=\Null(R(p))$ if and only if $\rank(R^\kappa(p))=\rank(R(p))$.
\end{proof}

The following result shows that bearing rigidity and global bearing rigidity are equivalent notions.

\begin{theorem}[Condition for Bearing Rigidity]\label{theorem_BRimplyGBR}
    A framework $\G(p)$ in $\R^d$ is bearing rigid if and only if it is globally bearing rigid.
\end{theorem}
\begin{proof}
    By definition, global bearing rigidity implies bearing rigidity.
    We next prove the converse is also true.
    Suppose the framework $\G(p)$ is bearing rigid.
    By the definition of bearing rigidity and Theorem~\ref{theorem_NSConditionForBEBC}, any framework satisfying $R(p)p'=0$ and $\|p'-p\|\le\epsilon$ also satisfies $R^\kappa(p)p'=0$, i.e.,
    $$R(p)(p+\delta p)=0 \Rightarrow R^\kappa(p)(p+\delta p)=0, \quad \forall \delta p, \|\delta p\|\le\epsilon,$$
    where $\delta p= p'-p$.
    It then follows from $R(p)p=0$ and $R^\kappa(p)p=0$ that
    $R(p)\delta p=0 \Rightarrow R^\kappa(p)\delta p=0$ for all $\|\delta p\|\le\epsilon.$
    This means $\Null(R(p))\subseteq \Null(R^\kappa(p))$ in spite of the constraint of $\|\delta p\|$.
    Since $\Null(R^\kappa(p))\subseteq \Null(R(p))$ as shown in Lemma~\ref{lemma_NullRkappa_NullR}, we further have $\Null(R(p))= \Null(R^\kappa(p))$ and consequently $\G(p)$ is globally bearing rigid.
\end{proof}

We next give the necessary and sufficient condition for infinitesimal bearing rigidity.

\begin{theorem}[Condition for Infinitesimal Bearing Rigidity]\label{theorem_conditionInfiParaRigid}
    For a framework $\G(p)$ in $\R^d$, the following statements are equivalent:
    \begin{enumerate}[(a)]
    \item $\G(p)$ is {infinitesimally bearing rigid};
    \item $\rank(R(p))=dn-d-1$;
    \item $\Null(R(p))=\myspan\{\one\otimes I_d, p\} = \myspan\{\one\otimes I_d, p-\one\otimes\bar{p}\}$, where $\bar{p}=(\one\otimes I_d)^\T p/n$ is the centroid of $\{p_i\}_{i\in\V}$.
    \end{enumerate}
\end{theorem}
\begin{proof}
Lemma~\ref{lemma_nullSpaceRigidity} shows $\myspan\{\one\otimes I_d, p\}\subseteq \Null(R(p))$.
Observe $\one\otimes I_d$ and $p$ correspond to a rigid-body translation and a scaling of the framework, respectively.
The stated condition directly follows from Definition \ref{definition_infinitesimalParallelRigid}.  Note also that $\{\one\otimes I_d, p-\one\otimes\bar{p}\}$ is an orthogonal basis for $\myspan\{\one\otimes I_d, p\}$.
\end{proof}

The special cases of $\R^2$ and $\R^3$ are of particular interest.
A framework $\G(p)$ is infinitesimally bearing rigid in $\R^2$ if and only if $\rank(R(p))=2n-3$, and in $\R^3$ if and only if $\rank(R(p))=3n-4$.
Note Theorem~\ref{theorem_conditionInfiParaRigid} does not require $n\ge d$.

The following result characterizes the relationship between infinitesimal bearing rigidity and global bearing rigidity.

\begin{theorem}\label{theorem_IBRimplyGBR}
Infinitesimal bearing rigidity implies global bearing rigidity.
\end{theorem}
\begin{proof}
Infinitesimal bearing rigidity implies $\Null(R(p))=\myspan\{\one\otimes I_d, p\}$.
Since $\myspan\{\one\otimes I_d, p\}\subseteq\Null(R^\kappa(p))\subseteq\Null(R(p))$ as shown in Lemma~\ref{lemma_NullRkappa_NullR}, it immediately follows from $\Null(R(p))=\myspan\{\one\otimes I_d, p\}$ that $\Null(R^\kappa(p))=\Null(R(p))$, which means $\G(p)$ is globally bearing rigid according to Theorem~\ref{theorem_NSConditionforGBR}.
\end{proof}

The converse of Theorem~\ref{theorem_IBRimplyGBR} is not true, i.e., global bearing rigidity does not imply infinitesimal bearing rigidity.
For example, the collinear framework as shown in Figure~\ref{fig_nonIBRExamples}(a) is globally bearing rigid but not infinitesimally bearing rigid.

We have at this point discussed three notions of bearing rigidity: (i)~bearing rigidity, (ii)~global bearing rigidity, and (iii)~infinitesimal bearing rigidity.
According to Theorem~\ref{theorem_BRimplyGBR} and Theorem~\ref{theorem_IBRimplyGBR}, the relationship between the three kinds of bearing rigidity can be summarized as below:
\begin{figure}[h]
\vskip-10pt
  \centering
  \def\myscale{0.35}
  \include{figure_tikz_relationBetweenTheRigidities}
  \vskip-20pt
\end{figure}

We next explore two important properties of infinitesimal bearing rigidity.
The following theorem shows that infinitesimal bearing rigidity can uniquely determine the shape of a framework.

\begin{theorem}[Unique Shape]\label{theorem_IPRImplyGPR}
    An infinitesimally bearing rigid framework can be uniquely determined up to a translational and a scaling factor.
\end{theorem}

\begin{proof}
    Suppose $\G(p)$ is an infinitesimally bearing rigid framework in $\R^d$.
    Consider an arbitrary framework $\G(p')$ that is bearing equivalent to $\G(p)$.
    Our aim is to prove $\G(p')$ is different from $\G(p)$ only in a translation and a scaling factor.
    The configuration $p'$ can always be decomposed as
    \begin{align}\label{eq_IPR_uniqueShape}
        p'=cp+\one\otimes \eta+q,
    \end{align}
    where $c\in\R\setminus\{0\}$ is the scaling factor, $\eta\in\R^d$ denotes a rigid-body translation of the framework, and $q\in\R^{dn}$, which satisfies $q\perp\myspan\{\one\otimes I_d,p\}$, represents a transformation other than translation and scaling.
    We only need to prove $q=0$.
    Since infinitesimal bearing rigidity implies that $\Null(R(p))=\myspan\{\one\otimes I_d, p\}$, multiplying $R(p)$ on both sides of \eqref{eq_IPR_uniqueShape} yields
    \begin{align}\label{eq_IPR_uniqueShape2}
        R(p)p'=R(p)q.
    \end{align}
    Since $\G(p')$ is bearing equivalent to $\G(p)$, we have $R(p)p'=0$ by Theorem~\ref{theorem_NSConditionForBEBC}.
    Therefore, \eqref{eq_IPR_uniqueShape2} implies $R(p)q=0$.
    Since $q\perp\myspan\{\one\otimes I_d,p\}=\Null(R(p))$, the above equation suggests $q=0$.
    As a result, $p'$ is different from $p$ only in a scaling factor $c$ and a rigid-body translation $\eta$.
\end{proof}

The following theorem shows that if a framework is infinitesimally bearing rigid in a lower dimension, it is still infinitesimally bearing rigid when evaluated in a higher dimensional space.

\begin{theorem}[Invariance to Dimension]\label{theorem_IBRInvariantDimension}
    Infinitesimal bearing rigidity is invariant to space dimensions.
\end{theorem}

\begin{proof}
Consider a framework $\G(p)$ in $\R^d$ ($n\ge2$, $d\ge2$).
Suppose the framework becomes $\G(\tilde{p})$ when the dimension is lifted from $d$ to $\tilde{d}$ ($\tilde{d}>d$).
Our goal is to prove that
\begin{align*}
    \rank(R(p))=dn-d-1 \Leftrightarrow \rank(R(\tilde{p}))=\tilde{d}n-\tilde{d}-1,
\end{align*}
and consequently $\G(\tilde{p})$ is infinitesimally bearing rigid in $\R^{\tilde{d}}$ if and only if $\G(p)$ is infinitesimally bearing rigid in $\R^d$.

First, consider an oriented graph and write the bearings of $\G(p)$ and $\G(\tilde{p})$ as $\{g_k\}_{k=1}^m$ and $\{\tilde{g}_k\}_{k=1}^m$, respectively.
Since $\tilde{p}_i$ is obtained from $p_i$ by lifting the dimension, without loss of generality, assume $\tilde{p}_i=[p_i^\T, 0]^\T$ ($\forall i\in\V$) where the zero vector is $(\tilde{d}-d)$-dimensional. Then,
\begin{align*}
    \tilde{g}_k=
    \left[
      \begin{array}{c}
        g_k \\
        0 \\
      \end{array}
    \right],
    \quad
    P_{\tilde{g}_k}=
    \left[
      \begin{array}{cc}
        P_{g_k} & 0\\
        0       & I_{\tilde{d}-d} \\
      \end{array}
    \right],\quad\forall k=\{1,\dots,m\}.
\end{align*}
The bearing rigidity matrix of $\G(\tilde{p})$ is $R(\tilde{p})=\dia{I_{\tilde{d}}/\|e_k\|}\dia{P_{\tilde{g}_k}}(H\otimes I_{\tilde{d}})$,
where
\begin{align*}
    &\dia{P_{\tilde{g}_k}}(H\otimes I_{\tilde{d}})\\
    &\qquad=\dia{
    \left[
      \begin{array}{cc}
        P_{g_k} & 0 \\
        0 & I_{\tilde{d}-d} \\
      \end{array}
    \right]}
    H\otimes
    \left[
      \begin{array}{cc}
        I_d & 0 \\
        0 & I_{\tilde{d}-d} \\
      \end{array}
    \right].
\end{align*}
Permutate the rows of $\dia{P_{\tilde{g}_k}}(H\otimes I_{\tilde{d}})$ to obtain
\begin{align*}
    A=
    \left[
      \begin{array}{c}
        \dia{P_{{g}_k}}H\otimes \left[
                                         \begin{array}{cc}
                                           I_d & 0 \\
                                         \end{array}
                                       \right] \\
        I_{(\tilde{d}-d)m}H\otimes \left[
                                         \begin{array}{cc}
                                           0 & I_{\tilde{d}-d} \\
                                         \end{array}
                                       \right] \\
      \end{array}
    \right]
    \triangleq
    \left[
      \begin{array}{c}
        A_1 \\
        A_2 \\
      \end{array}
    \right].
\end{align*}
Since the permutation of the rows does not change the matrix rank, we have $\rank(R(\tilde{p}))=\rank(A)$.
Because the rows of $A_1$ are orthogonal to the rows of $A_2$, we have $\rank(A)=\rank(A_1)+\rank(A_2)$.
As a result, considering $\rank(A_1)=\rank(\dia{P_{{g}_k}}H\otimes I_d)=\rank(R(p))$ and $\rank(A_2)=\rank(H\otimes I_{\tilde{d}-d})=(\tilde{d}-d)(n-1)$,
we have
\begin{align*}
    \rank(R(\tilde{p}))
    %&=\rank(A)=\rank(A_1)+\rank(A_2)\\
    &=\rank(R(p))+(\tilde{d}-d)(n-1).
\end{align*}
It can be easily verified using the above equation that $\rank(R(\tilde{p}))=\tilde{d}n-\tilde{d}-1$ if and only if $\rank(R(p))=dn-d-1$.
\end{proof}

\begin{figure}[!t]
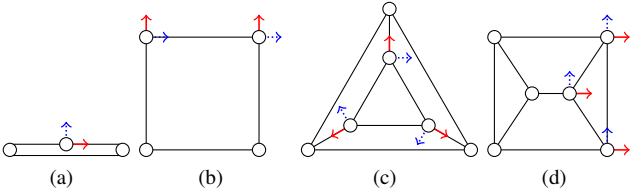

  \centering
  \def\myscale{0.5}
  \include{fig_tikz_Example_nonIBR}
  \caption{Examples of \emph{non-infinitesimally bearing rigid} frameworks. The red arrows (solid) stand for non-trivial infinitesimal bearing motions and the blue arrows (dashed) for the associated orthogonal infinitesimal distance motions.}
  \label{fig_nonIBRExamples}
  \vspace{-12pt}
\end{figure}

Figure~\ref{fig_nonIBRExamples} shows examples of non-infinitesimal bearing rigid frameworks.
The frameworks in Figure~\ref{fig_nonIBRExamples} are not infinitesimally bearing rigid because there exist {non-trivial} infinitesimal bearing motions (see, for example, the red arrows).
Figure~\ref{fig_IPRExamples} shows some two- and three-dimensional infinitesimally bearing rigid frameworks.
It can be verified that each of the frameworks satisfies $\rank(R(p))=dn-d-1$.

\begin{figure}[!t]
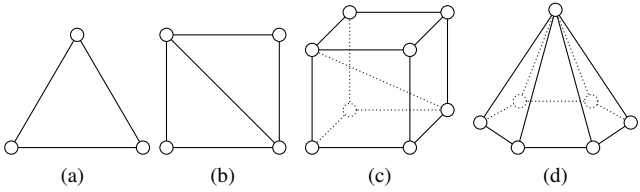

  \centering
  \def\myscale{0.5}
  \include{figure_tikz_IPRExamples}
  \caption{
  Examples of \emph{infinitesimally bearing rigid} frameworks.
  }
  \label{fig_IPRExamples}
  \vspace{-10pt}
\end{figure}
\subsection{Connections to Distance Rigidity Theory}

The bearing rigidity theory and the distance rigidity theory study similar problems of whether the shape of a framework can be uniquely determined by the inter-neighbor bearings and inter-neighbor distances, respectively.
It is meaningful to study the connections between the two rigidity theories.
The following theorem establishes the equivalence between infinitesimal bearing and distance rigidity in $\R^2$.

\begin{theorem}\label{theorem_2DIPRandIREquivalence}
    In $\R^2$, a framework is infinitesimally bearing rigid if and only if it is infinitesimally distance rigid.
\end{theorem}
\begin{proof}
See Appendix~\ref{appendix_distanceRigidity}.
\end{proof}

Two remarks on Theorem~\ref{theorem_2DIPRandIREquivalence} are given below.
Firstly, Theorem~\ref{theorem_2DIPRandIREquivalence} cannot be generalized to $\R^3$ or higher dimensions.
For example, the three-dimensional cubic and hexagonal pyramid frameworks in Figure~\ref{fig_IPRExamples}(c)-(d) are infinitesimally bearing rigid but not distance rigid.
In particular, the rank of the distance rigidity matrices of the two frameworks are $13$ and $12$, respectively, whereas the required ranks for infinitesimal distance rigidity are $18$ and $15$, respectively.
Secondly, Theorem~\ref{theorem_2DIPRandIREquivalence} suggests that we can determine the infinitesimal distance rigidity of a framework by examining its infinitesimal bearing rigidity.
For example, it may be tricky to see the frameworks in Figure~\ref{fig_nonIBRExamples}(c)-(d) are not infinitesimally distance rigid, but it is obvious to see the non-trivial infinitesimal bearing motions and conclude they are not infinitesimally bearing rigid.

%The relationship between infinitesimal bearing and distance motions of a framework is now discussed.
An immediate corollary of Theorem \ref{theorem_2DIPRandIREquivalence} describes the relationship between infinitesimal bearing motions and infinitesimal distance motions of frameworks in $\R^2$.
Let $Q_{\pi/2}\in\mathcal{SO}(2)$ be a rotation matrix that can rotate a vector in $\R^2$ by $\pi/2$.
For any $\delta p=[\delta p_1^\T,\dots,\delta p_n^\T]^\T\in\R^{2n}$, denote $\delta p^\perp=[(Q_{\pi/2}\delta p_1)^\T,\dots,(Q_{\pi/2}\delta p_n)^\T]^\T\in\R^{2n}$.

\begin{corollary}\label{corollary_infinitesimalBearngAndDistanceMotion}
The vector $\delta p$ is an infinitesimal bearing motion of a framework $\G(p)$ in $\R^2$ if and only if $\delta p^\perp$ is an infinitesimal distance motion of $\G(p)$.
\end{corollary}
\begin{proof}
See Appendix~\ref{appendix_distanceRigidity}.
\end{proof}

Given a framework in $\R^2$, Corollary~\ref{corollary_infinitesimalBearngAndDistanceMotion} suggests that for any infinitesimal bearing motion, there exists a perpendicular infinitesimal distance motion, and the converse is also true.
Corollary~\ref{corollary_infinitesimalBearngAndDistanceMotion} is illustrated by Figure~\ref{fig_nonIBRExamples} (indicated by the red (solid) and blue (dashed) arrows).

To end this section, we briefly compare the proposed bearing rigidity theory with the well-known distance rigidity theory.
In the distance rigidity theory, there are three kinds of rigidity: (i)~distance rigidity, (ii)~global distance rigidity, and (iii)~infinitesimal distance rigidity.
The relationship between them is (ii)$\Rightarrow$(i) and (iii)$\Rightarrow$(i).
Note (ii) and (iii) do not imply each other.
The global distance rigidity can uniquely determine the shape of a framework, but it is usually difficult to mathematically examine \cite{Hendrickson1992SIAM,Connelly2005Generic}.
Infinitesimal distance rigidity can be conveniently examined by a rank condition (see Lemma~\ref{lemma_diatanceInfiRigid_NSCondition} in Appendix~\ref{appendix_distanceRigidity}), but it is not able to ensure a unique shape.
As a comparison, the proposed infinitesimal bearing rigidity not only can be examined by a rank condition (Theorem~\ref{theorem_conditionInfiParaRigid}) but also can ensure the unique shape of a framework (Theorem~\ref{theorem_IPRImplyGPR}).
In addition, the rank condition for infinitesimal distance rigidity requires to distinguish the cases of $n\ge d$ and $n<d$ (Lemma~\ref{lemma_diatanceInfiRigid_NSCondition}), while the rank condition for infinitesimal bearing rigidity does not.
Finally, an infinitesimally distance rigid framework in a lower dimension may become non-rigid in a higher dimension (see, for example, Figure~\ref{fig_IPRExamples}(b)), while infinitesimal bearing rigidity is invariant to dimensions.
In summary, the bearing rigidity theory possesses a number of attractive features compared to the distance rigidity theory, and as we will show in the sequel, it is a powerful tool for analyzing bearing-based formation control problems.

\section{Bearing-only Formation Control with a Global Reference Frame}\label{section_BOF_global}

In this section, we study bearing-only formation control of multi-agent systems in arbitrary dimensions in the presence of a global reference frame.
Consider $n$ agents in $\R^d$ ($n\ge2$ and $d\ge2$).
Note $n\ge d$ is not required.
Assume there is a global reference frame known to each agent.
All the vector quantities given in this section are expressed in this global frame.
Denote ${p}_i\in\R^d$ as the position of agent $i\in\{1,\dots,n\}$.
The dynamics of agent $i$ is
\begin{align*}
    \dot{p}_i(t) = v_i(t),
\end{align*}
where $v_i(t)\in\R^d$ is the velocity input to be designed.
Denote $p=[p_1^\T ,\dots,p_n^\T ]^\T \in\R^{dn}$ and $v=[v_1^\T,\dots,v_n^\T ]^\T \in\R^{dn}$.
The underlying sensing graph $\G=(\V,\E)$ is assumed to be undirected and fixed, and the formation is denoted by $\G(p)$.
The {edge vector} $e_{ij}$ and the bearing $g_{ij}$ are defined as in \eqref{eq_eij_gij_def}.
Considering an arbitrary orientation of $\G$, we can reexpress the edge and bearing vectors as $e=[e_1^\T ,\dots,e_m^\T ]^\T$ and $g=[g_1^\T ,\dots,g_m^\T ]^\T$ as defined in \eqref{eq_ek_gk_def}.

If $(i,j)\in\E$, agent $i$ can measure the relative bearing $g_{ij}$ of agent $j$.
As a result, the \emph{bearing measurements} obtained by agent $i$ at time $t$ are $\left\{g_{ij}(t)\right\}_{j\in\N_i}$.
The constant \emph{bearing constraints} for the target formation are specified as $\{g_{ij}^*\}_{(i,j)\in\E}$ with $g_{ij}^*=-g_{ji}^*$.
Figure~\ref{fig_simExampleDemo} gives two examples to illustrate the bearing constraints.

\begin{figure}
  \centering
  \subfloat[]{\includegraphics[width=0.35\linewidth]{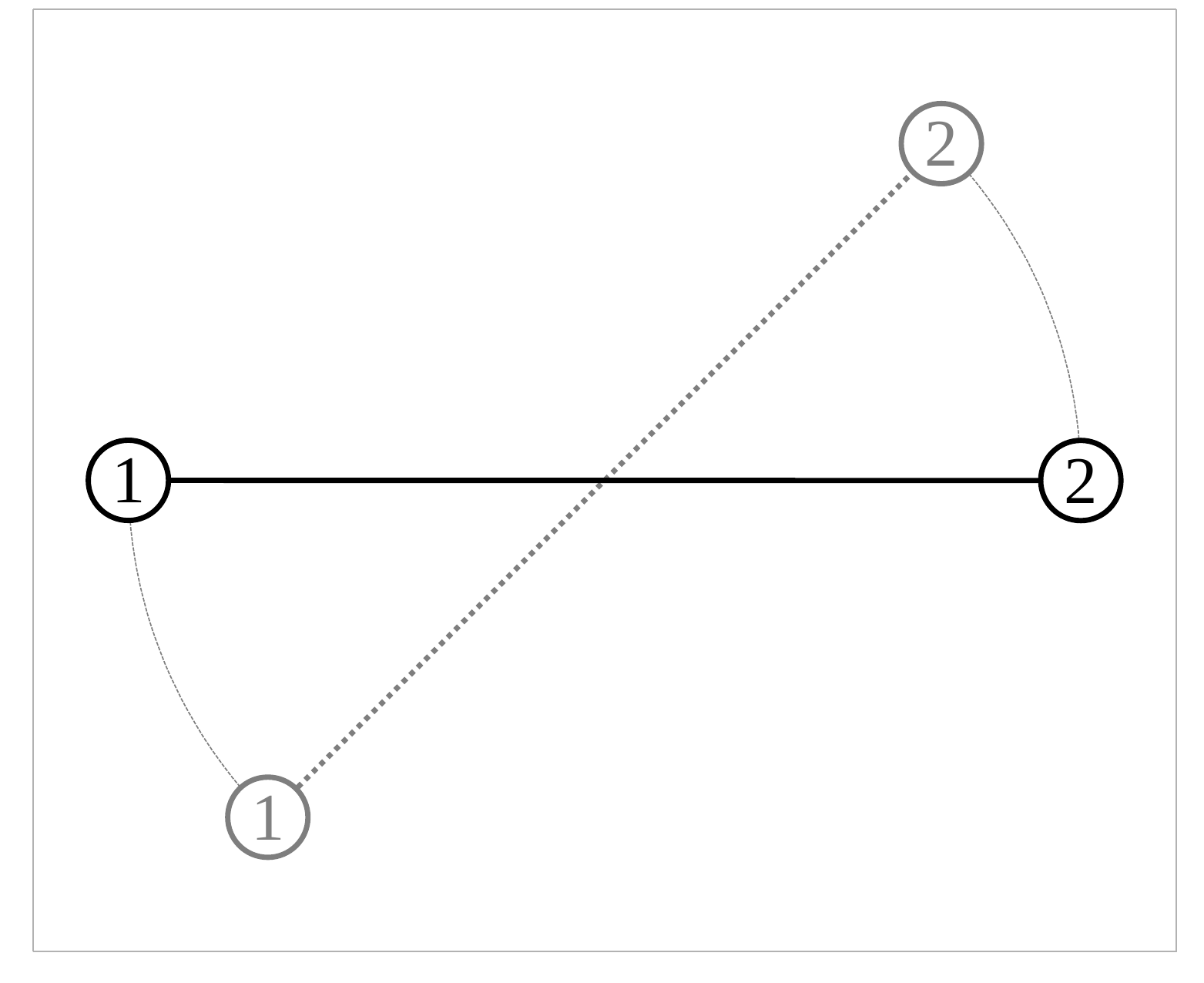}}
  \subfloat[]{\includegraphics[width=0.35\linewidth]{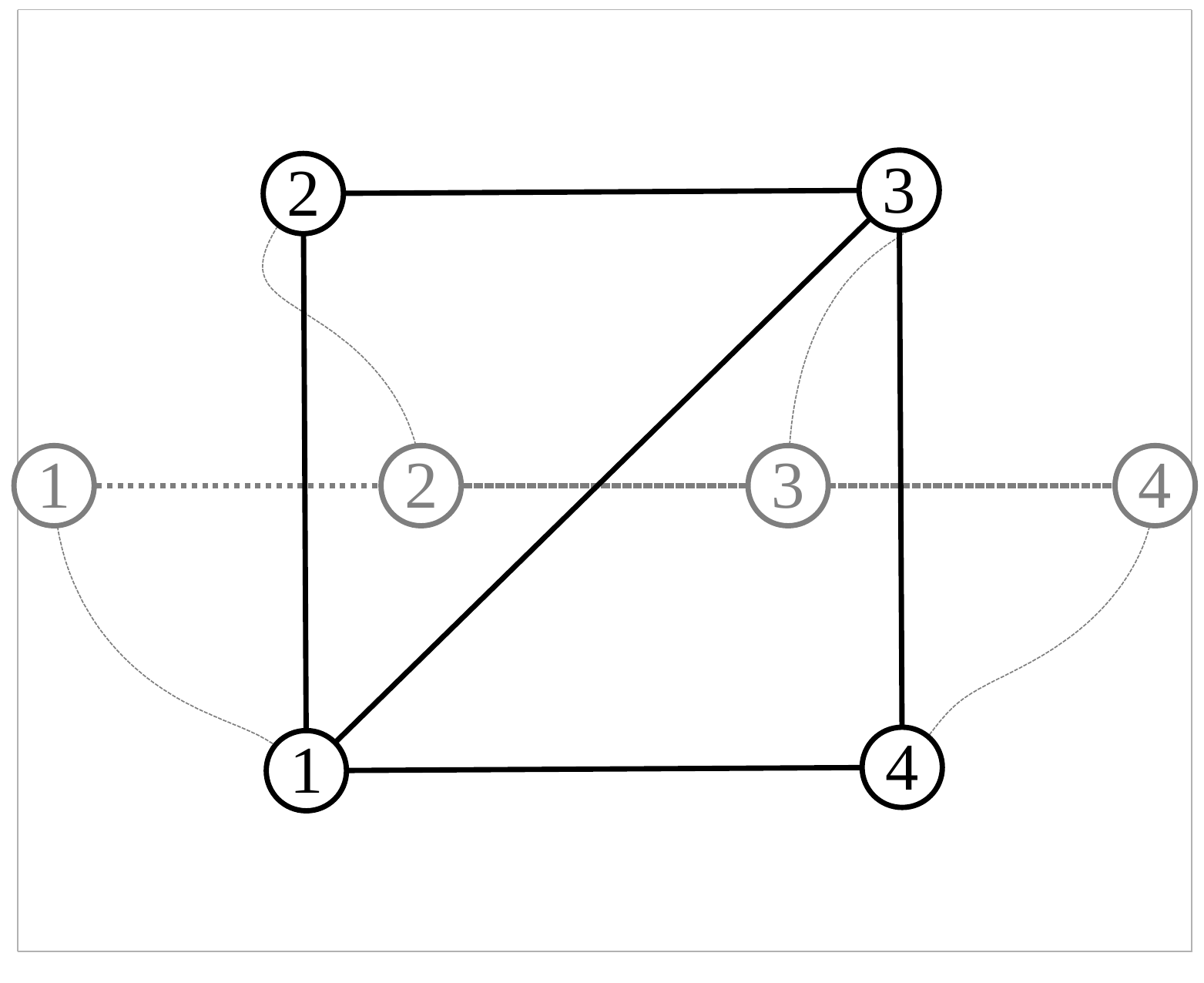}}
  \caption{%
  Target formation in black; initial formation in grey.
  (a) Bearing constraints for the target formation are $g_{12}^*=-g_{21}^*=[1,0]^\T$.
  (b) Bearing constraints for the target formation are
  $g_{12}^*=-g_{21}^*=[0,1]^\T$, $g_{23}^*=-g_{32}^*=[1,0]^\T$, $g_{34}^*=-g_{43}^*=[0,-1]^\T$, $g_{41}^*=-g_{14}^*=[-1,0]^\T$, and $g_{13}^*=-g_{31}^*=[\sqrt{2}/2,\sqrt{2}/2]^\T$.}
  \label{fig_simExampleDemo}
  \vspace{-10pt}
\end{figure}

\begin{definition}[Feasible Bearing Constraints]
    The bearing constraints $\{g_{ij}^*\}_{(i,j)\in\E}$ are \emph{feasible} if there exists a formation $\G(p)$ that satisfies $g_{ij}=g_{ij}^*$ for all $(i,j)\in\E$.
\end{definition}

Feasible bearing constraints can be easily calculated from an arbitrary configuration that has the desired geometric formation pattern.
The bearing-only formation control problem to be solved in this section is formally stated below.

\begin{problem}\label{problem_global_bearingonlyformationcontrol}
    Given feasible constant bearing constraints $\{g_{ij}^*\}_{(i,j)\in\E}$ and the initial formation $\G(p(0))$, design $v_i(t)$ for agent $i\in\V$ based only on the bearing measurements $\left\{g_{ij}(t)\right\}_{j\in\N_i}$ such that $g_{ij}(t)\rightarrow g_{ij}^*$ as $t\rightarrow\infty \, \forall \,(i,j)\in\E$.
\end{problem}

\subsection{A Bearing-Only Control Law}
The proposed nonlinear bearing-only control law is
\begin{align}\label{eq_controlLawElement}
    v_i(t) = - \sum_{j\in\N_i} P_{g_{ij}(t)}g_{ij}^*, \quad i\in\V
\end{align}
where $P_{g_{ij}(t)}=I_d-g_{ij}(t)g_{ij}^\T(t)$.
First, the control law is distributed since the control of agent $i$ only requires the bearing measurements $\{g_{ij}(t)\}_{j\in\N_i}$.
Second, the control input is always bounded as $\|v_i(t)\|\le|\N_i|$ since $\|P_{g_{ij}(t)}\|=\|g_{ij}^*\|=1$.
Third, the control law has a clear geometric interpretation as illustrated in Figure~\ref{fig_controlLawGeometricMeaning}: the control term $-P_{g_{ij}}g_{ij}^*$ is perpendicular to $g_{ij}$ since $g_{ij}^\T P_{g_{ij}}g_{ij}^*=0$.
As a result, the control law attempts to reduce the bearing error of $g_{ij}$ while preserving the distance between agents $i$ and $j$.
This geometric interpretation can also be demonstrated by the example shown in Figure~\ref{fig_simExampleDemo}(a), where the bearing error is reduced to zero while the inter-agent distance is preserved.
In addition, similar ``projective'' control laws have been used before in \cite{Moshtagh2007CDC,Justh2005CDC} for circular formation coordination control.

\begin{figure}
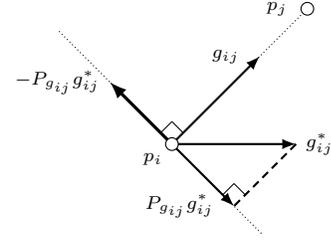

  \centering
  \def\myscale{0.30}
  \include{figure_tikz_controlLawMeaning}
  \vspace{-8pt}
  \caption{The geometric interpretation of control law \eqref{eq_controlLawElement}. The control term $-P_{g_{ij}}g_{ij}^*$ is perpendicular to the bearing $g_{ij}$.}
  \label{fig_controlLawGeometricMeaning}
  \vspace{-13pt}
\end{figure}

In order to analyze the proposed control law, we rewrite it in a matrix-vector form.
Since $g_{ij}^*=-g_{ji}^*$, the bearing constraints $\{g_{ij}^*\}_{(i,j)\in\E}$ can be reexpressed as $\{g_k^*\}_{k=1}^m$ by considering an oriented graph.
Let $g^*=[(g_1^*)^\T ,\dots,(g_m^*)^\T ]^\T$, then \eqref{eq_controlLawElement} can be written as
\begin{align}\label{eq_controlLawMatrix}
    v = \bar{H}^\T\mydiag(P_{g_k}) g^* \triangleq \tilde{R}^\T (p)g^*.
\end{align}
It should be noted that the oriented graph is merely used to obtain the matrix expression while the underlying sensing graph of the formation is still the undirected graph $\G$.
Moreover, it is worth mentioning that control law \eqref{eq_controlLawMatrix} is a modified gradient control law. If we consider the bearing error $\sum_{k=1}^m \|g_k-g_k^*\|^2$, a short calculation shows the corresponding gradient control law is $u=\bar{H}^\T \mydiag({P_{g_k}}/{\|e_k\|})g^*$, which is exactly $u=R^\T (p)g^*$, where $R(p)$ is the bearing rigidity matrix.
This gradient control law, however, requires the distance measurement $\|e_k\|$.
By removing the distance term $\|e_k\|$, we can obtain the proposed control law \eqref{eq_controlLawMatrix}.

We next examine some useful properties of the control law.
First of all, we show that both the centroid and scale of the formation are invariant quantities under the action of the control law.
In this direction, define the \emph{centroid} and \emph{scale} of the formation as
\begin{align*}
    \bar{p} \triangleq \frac{1}{n}\sum_{i=1}^n p_i, &\quad  s  \triangleq \sqrt{\frac{1}{n}\sum_{i=1}^n \|p_i-\bar{p}\|^2}.
\end{align*}
The scale is the quadratic mean of the distances from the agents to the centroid.

\begin{lemma}\label{lemma_xDotPerpendicular}
    Under control law \eqref{eq_controlLawMatrix}, $\dot{p}(t)\perp \myspan\left\{\one\otimes I_d, p(t)\right\}$.
\end{lemma}
\begin{proof}
The dynamics $\dot{p}=\tilde{R}^\T (p)g^*$ implies $\dot{p}\in\Range(\tilde{R}^\T (p))$.
Since $\Range(\tilde{R}^\T (p))\perp\Null(\tilde{R}(p))$, we have $\dot{p}\perp\Null(\tilde{R}(p))$.
Furthermore, $\Null(\tilde{R}(p))=\Null(R(p))$ and $\myspan\{\one\otimes I_d, p\}\subseteq \Null(R(p))$ by  Lemma~\ref{lemma_nullSpaceRigidity} concludes the proof.
\end{proof}

\begin{theorem}[Centroid and Scale Invariance]\label{propos_invariantCentroidScale}
    The centroid $\bar{p}$ and the scale $s$ are invariant under the control law \eqref{eq_controlLawMatrix}.
\end{theorem}
\begin{proof}
    Since $\bar{p}=(\one\otimes I_d)^\T  p/n$, we have $\dot{\bar{p}}=(\one\otimes I_d)^\T  \dot{p}/n$.
    It follows from $\dot{p}\perp\Range(\one\otimes I_d)$ as shown in Lemma~\ref{lemma_xDotPerpendicular} that $\dot{\bar{p}}\equiv0$.
    Rewrite $s$ as $s=\|p-\one\otimes \bar{p}\|/\sqrt{n}$.
    Then,
    \begin{align*}
        \dot{s}
        &=\frac{1}{\sqrt{n}} \frac{(p-\one\otimes \bar{p})^\T }{\|p-\one\otimes \bar{p}\|} \dot{p}.
    \end{align*}
    It follows from $\dot{p}\perp{p}$ and $\dot{p}\perp\one\otimes\bar{p}$ that $\dot{s}\equiv0$.
\end{proof}

Theorem~\ref{propos_invariantCentroidScale} can be well demonstrated by the simple simulation example as shown in Figure~\ref{fig_simExampleDemo}(a).
As can be seen, the middle point (i.e., the centroid) and the distance of the two agents (i.e., the scale) are invariant during the formation evolution.
The invariance of centroid and scale has also been observed by \cite{Eric2014ACC} for bearing-only formation control in two-dimensional cases.

The following results, which can be obtained from Theorem~\ref{propos_invariantCentroidScale}, characterize the behavior of the formation trajectories.
In particular, the bounds for the quantities $\max_{i\in\V}\|p_i(t)-\bar{p}\|$ and $\|p_i(t)-p_j(t)\|, \forall i,j\in\V$ are given.

\begin{corollary}\label{corollary_maximumpipbarDistance}
The formation trajectory under the control law \eqref{eq_controlLawMatrix} satisfies the following inequalities,
\begin{enumerate}[(a)]
	\item $s\le \max_{i\in\V}\|p_i(t)-\bar{p}\| \le s\sqrt{n-1}, \quad \forall t\ge0.$
	\item $ \|p_i(t)-p_j(t)\|\le 2s\sqrt{n-1},\quad\forall i,j\in\V,\ \forall t\ge0.$
\end{enumerate}
\end{corollary}
\begin{proof}
We first prove $\|p_i-\bar{p}\| \le s\sqrt{n-1}$ for all $i\in\V$.
On one hand, $\sum_{j\in\V}(p_j-\bar{p})=(p_i-\bar{p})+\sum_{j\in\V,j\ne i}(p_j-\bar{p})=0$ implies
\begin{align}\label{eq_pipbarDistance_inequality1}
    \hspace{-5pt}\|p_i-\bar{p}\|^2
    &\le \left(\sum_{\substack{j\in\V \\j\ne i}}\|p_j-\bar{p}\|\right)^2 \hspace{-5pt}\le (n-1)\hspace{-3pt}\sum_{\substack{j\in\V,\\j\ne i}}\|p_j-\bar{p}\|^2.
\end{align}
On the other hand, scale invariance implies that $\|p_i-\bar{p}\|^2+\sum_{j\in\V,j\ne i}\|p_j-\bar{p}\|^2=ns^2$.
Substituting this expression into \eqref{eq_pipbarDistance_inequality1} gives
$\|p_i-\bar{p}\|^2\le (n-1)(ns^2-\|p_i-\bar{p}\|^2)$,
which implies $\|p_i-\bar{p}\| \le s\sqrt{n-1}$.
We next prove $s\le\max_{i\in\V}\|p_i-\bar{p}\|$.
Since $\max_{i\in\V}\|p_i-\bar{p}\|^2\ge\|p_j-\bar{p}\|^2$, we have
$
    n(\max_{i\in\V} \|p_i-\bar{p}\|^2)\ge\sum_{i=1}^n\|p_i-\bar{p}\|^2=ns^2,
$
which implies $\max_{i\in\V} \|p_i-\bar{p}\|\ge s$.
The inequality in (b) is obtained from
$\|p_i(t)-p_j(t)\|=\|(p_i(t)-\bar{p})-(p_j(t)-\bar{p})\|\le\|p_i(t)-\bar{p}\|+\|p_j(t)-\bar{p}\|\le 2s\sqrt{n-1}$.
\end{proof}

\subsection{Formation Stability Analysis}

In order to prove the formation stability, we adopt the following rigidity assumption.
\begin{assumption}\label{assumption_parallelRigid}
    The bearing constraints $\{g_{ij}^*\}_{(i,j)\in\E}$ ensures infinitesimal bearing rigidity.
\end{assumption}

Assumption~\ref{assumption_parallelRigid} gives two conditions that will be useful for the formation stability analysis.
The first condition is that the shape of any formation that satisfies the bearing constraints is unique according to Theorem~\ref{theorem_IPRImplyGPR}.
The second condition is a mathematical condition.
More specifically, suppose $\G(p)$ is a formation that satisfies the bearing constraints, then Assumption~\ref{assumption_parallelRigid} indicates that the bearing rigidity matrix $R(p)$ satisfies $\rank(R(p))=dn-d-1$ and $\Null(R(p))=\myspan\{\one\otimes I_d, p\}$ according to Theorem~\ref{theorem_conditionInfiParaRigid}.

The basic idea of the formation stability proof is to show that the formation converges from an initial formation $\G(p(0))$ to a \emph{target formation} $\G(p^*)$ as defined below.

\begin{definition}[Target Formation]\label{definition_globalcase_targetFormation}
    Let $\G(p^*)$ be a \emph{target formation} satisfying
    \begin{enumerate}[(a)]
        \item Centroid: $\bar{p}^*=\bar{p}(0)$.
        \item Scale: $s^*=s(0)$.
        \item Bearing: $(p_j^*-p_i^*)/\|p_j^*-p_i^*\|=g_{ij}^*$ for all $(i,j)\in\E$.
    \end{enumerate}
\end{definition}

The target formation $\G(p^*)$ has the same centroid and scale as the initial formation and it satisfies all the bearing constraints.

\begin{lemma}[Existence and Uniqueness]\label{lemma_targetFormationExistUnique}
The target formation $\G(p^*)$ in Definition~\ref{definition_globalcase_targetFormation} always exists and is unique under Assumption~\ref{assumption_parallelRigid}.
\end{lemma}
\begin{proof}
Since the bearing constraints are feasible, there exist formations that satisfy the bearings. Due to the infinitesimal bearing rigidity in Assumption~\ref{assumption_parallelRigid}, these formations including $\G(p^*)$ can be uniquely determined up to translations and scaling factors.
Since $\G(p^*)$ additionally has the centroid and the scale as $\bar{p}(0)$ and $s(0)$, the translation and the scale of $\G(p^*)$ can be uniquely determined.
\end{proof}
\begin{remark}
The unique value of $p^*$ can be calculated as below.
From the bearing constraints, construct $\tilde{R}\triangleq\mydiag(P_{g_k^*})\bar{H}$, which has the same null space as the bearing rigidity matrix $R(p^*)$.
It follows from the infinitesimal bearing rigidity that $\Null(\tilde{R})=\Null(R(p^*))=\myspan\{\one\otimes I_d, p^*\}$.
Suppose $\myspan\{\one\otimes I_d, q\}$ is an orthogonal basis of $\Null(\tilde{R})$ obtained by numerical calculation.
Since $p^*\in\Null(\tilde{R})$, we can express $p^*$ as a linear combination of $\one\otimes I_d$ and $q$, $p^*=\one\otimes x+\alpha q,$
where $x\in\R^d$ and $\alpha \in\R$ are the coefficients to be calculated.
Since $\bar{p}^*=(\one\otimes I_d)^\T p^*/n=\bar{p}(0)$ and $s^*=\|p^*-\one\otimes \bar{p}^*\|/\sqrt{n}=s(0)$, a short calculation shows that $x=\bar{p}(0)$ and $\alpha =\pm s(0)\sqrt{n}/\|q\|$. The correct sign of $\alpha $ can be determined by comparing the signs of $q_j-q_i$ and $g_{ij}^*$.
\end{remark}

The stability proof is to show that the formation converges to the target formation and consequently the bearing errors converge to zero.
This idea was originally proposed by \cite{Eric2014ACC} to solve bearing-only formation control in two dimensions.
In this direction, let $\delta_i= p_i-p_i^*$ and then $\dot{\delta}_i=f_i(\delta)=\dot{p}_i$.
Denote $\delta=[\delta_1^\T,\dots,\delta_n^\T]^\T$ and $f(\delta)=[f_1^\T (\delta),\dots,f_n^\T (\delta)]^\T$.
With control law \eqref{eq_controlLawMatrix}, the $\delta$-dynamics is expressed as
\begin{align}\label{eq_Global_deltaDynamics}
    \dot{\delta}(t) = f(\delta)=\bar{H}^\T\mydiag(P_{g_k})g^*.
\end{align}
Our aim is to show $\delta(t)$ converges to zero.
We next identify the equilibriums of the $\delta$-dynamics.
Denote
$$r(t)\triangleq p(t)-(\one\otimes\bar{p}), \quad r^*\triangleq p^*-(\one\otimes\bar{p}^*).$$
Note $r(t)$ is obtained by moving the centroid of $p(t)$ to the origin.
Due to the scale invariance, it can be verified that $\|r(t)\|\equiv\|r^*\|=\sqrt{n}s$ for all $t\ge0$.
Moreover, since $\bar{p}=\bar{p}^*$, we have $\delta(t)=r(t)-r^*$.

\begin{figure}
  \centering
  \includegraphics[width=0.38\linewidth]{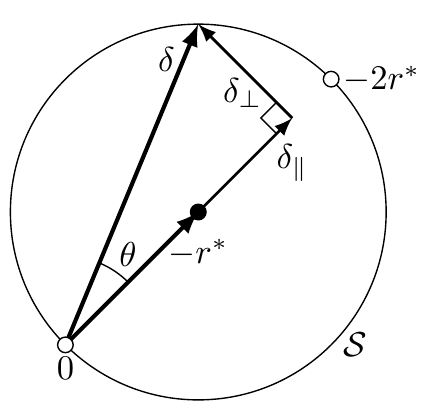}
  \caption{Geometric interpretation of $\delta$ which satisfies $\|\delta+r^*\|=\|r^*\|$.}
  \label{fig_deltaGeometry}
  \vspace{-10pt}
\end{figure}

\begin{lemma}
    System \eqref{eq_Global_deltaDynamics} evolves on the surface of the sphere
    \begin{align*}
        \S=\{\delta\in\R^{dn}: \|\delta+r^*\|=\|r^*\|\}.
    \end{align*}
\end{lemma}

\begin{proof}
It follows from $\delta(t)=r(t)-r^*$ that $\|\delta(t)+r^*\|=\|r(t)\|=\|r^*\|$, where $\|r(t)\|=\|r^*\|$ is due to the scale invariance.
\end{proof}

The state manifold $\S$ is illustrated by Figure~\ref{fig_deltaGeometry}.
We next introduce a useful lemma and then prove that system \eqref{eq_Global_deltaDynamics} has \emph{two isolated equilibriums} on $\S$.

\begin{lemma}\label{lemma_exchangeP1g2}
    Any two unit vectors $g_1, g_2\in\R^d$ always satisfy $g_1^\T P_{g_2}g_1=g_2^\T P_{g_1}g_2$.
\end{lemma}
\begin{proof}
Since $g_1^\T g_1=g_2^\T g_2=1$, we have $g_1^\T P_{g_2}g_1=g_1^\T (I_d-g_2g_2^\T )g_1=g_1^\T g_1-g_1^\T g_2g_2^\T g_1=g_2^\T g_2-g_2^\T g_1g_1^\T g_2=g_2^\T (I_d-g_1g_1^\T )g_2=g_2^\T P_{g_1}g_2$.
\end{proof}

\begin{theorem}[Equilibrium]\label{theorem_global_twoEquilibrium}
    Under Assumption~\ref{assumption_parallelRigid}, system \eqref{eq_Global_deltaDynamics} has two isolated equilibriums, $\delta=0$ and $\delta=-2r^*$.
\end{theorem}
\begin{proof}
Any equilibrium $\delta\in\S$ must satisfy $f(\delta)=\bar{H}^\T\mydiag(P_{g_k})g^*=0$, which implies
\begin{align*}
    0
    &=(p^*)^\T\bar{H}^\T\mydiag(P_{g_k})g^*
    =(e^*)^\T\mydiag(P_{g_k})g^* \\
    &=\sum_{k=1}^m (e^*_k)^\T P_{g_k}g_k^*
    =\sum_{k=1}^m \|e_k^*\|(g^*_k)^\T P_{g_k}g_k^*.
\end{align*}
Since $(g^*_k)^\T P_{g_k}g_k^*\ge0$, the above equation implies $(g^*_k)^\T P_{g_k}g_k^*=0$ for all $k$.
As a result, by Lemma~\ref{lemma_exchangeP1g2}, we have $g_k^\T  P_{g_k^*}g_k=0\Rightarrow e_k^\T  P_{g_k^*}e_k=0$ for all $k$ and thus
\begin{align*}
    0& = e^\T \dia{P_{g_k^*}}e = p^\T \underbrace{\bar{H}^\T\dia{P_{g_k^*}}}_{\tilde{R}^\T (p^*)}\underbrace{\dia{P_{g_k^*}}\bar{H}}_{\tilde{R}(p^*)}p,
\end{align*}
where the last equality is due to the facts that $P_{g_k^*}=P_{g_k^*}^2$ and $e=\bar{H}p$.
The above equation indicates $\tilde{R}(p^*)p=0.$
%\begin{align*}
%    \tilde{R}(p^*)p=0.
%\end{align*}
    Observe $\tilde{R}(p^*)=\mydiag(P_{g^*_k})\bar{H}$ has the same null space as the bearing rigidity matrix $R(p^*)=\mydiag(P_{g^*_k}/\|e^*_k\|)\bar{H}$.
    Since $\G(p^*)$ is infinitesimally bearing rigid by Assumption~\ref{assumption_parallelRigid}, it follows from Theorem~\ref{theorem_conditionInfiParaRigid} that
    $\Null(\tilde{R}(p^*))= \myspan\{\one\otimes I_d,p^*-\one\otimes\bar{p}^*\}$.
    Considering $\tilde{R}(p^*)p=0\Leftrightarrow\tilde{R}(p^*)(p-\one\otimes\bar{p})=0$, we have
    \begin{align*}
        p-\one\otimes\bar{p}\in\myspan\{\one\otimes I_d, p^*-\one\otimes\bar{p}^*\}.
    \end{align*}
    Because $p-\one\otimes\bar{p}\perp \Range(\one\otimes I_d)$, we further know
    $p-\one\otimes\bar{p}\in\myspan\{p^*-\one\otimes\bar{p}^*\}$.
    Moreover, since $\|p-\one\otimes\bar{p}\|=\|p^*-\one\otimes\bar{p}^*\|$ due to the scale invariance, we have
    \begin{align*}
         p-\one\otimes\bar{p}= \pm(p^*-\one\otimes\bar{p}^*).
    \end{align*}
    (i)~In the case of $p-\one\otimes\bar{p}=p^*-\one\otimes\bar{p}^*$, we have $p=p^*\Leftrightarrow\delta=0$ and consequently $g_{ij}=g_{ij}^*$ for all $(i,j)\in\E$.
   (ii)~In the case of $p-\one\otimes\bar{p} = -(p^*-\one\otimes\bar{p}^*)$, we have $p=-p^*+2(\one\otimes\bar{p}^*)\Leftrightarrow\delta=-2(p^*-\one\otimes\bar{p}^*)$, and consequently $g_{ij}=-g_{ij}^*$ for all $(i,j)\in\E$.
\end{proof}

\begin{figure}[t]
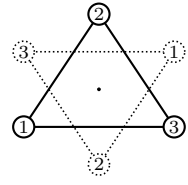

  \centering
  \def\myscale{0.5}
  \include{figure_tikz_undesiredEuqilibrium}
  \vspace{-10pt}
  \caption{An illustration of the target formation where $\delta =0$ (solid) and its corresponding point reflection where $\delta=-2r^*$ (dashed).}
  \label{fig_twoEquilibriums}
  \vspace{-10pt}
\end{figure}

The equilibrium $\delta=0$ is desired whereas the other one $\delta=-2r^*$ is undesired.
As shown in the proof, the formation at $\delta=-2r^*$ is geometrically a \emph{point reflection} of the target formation about the \emph{centroid}.
As a result, the two formations at the two equilibriums have the same centroid, scale, and shape, but they have the opposite bearings.
See Figure~\ref{fig_twoEquilibriums} for an illustration.

Although we will present a nonlinear stability analysis of the two equilibriums later, it is still meaningful to examine the Jacobian matrices at the two equilibriums.
Based on the Jacobian matrices, we are able to conclude by Lyapunov's indirect method that the undesired equilibrium $\delta=-2r^*$ is unstable.

\begin{proposition}\label{proposition_global_Jacobian}
Let $A={\partial f(\delta)}/{\partial \delta}$ be the Jacobian of $f(\delta)$.
At the undesired equilibrium $\delta=-2r^*$, the Jacobian matrix $\left.A\right|_{\delta=-2r^*}$ is symmetric positive semi-definite and at least one eigenvalue is positive.
As a result, the undesired equilibrium $\delta=-2r^*$ is unstable.
\end{proposition}

\begin{proof}
    Recall $f_i(\delta)=-\sum_{j\in\N_i} P_{g_{ij}}g_{ij}^*, \forall i\in\V$.
    For any $j\notin\N_i$, we have $A_{ij}=\partial f_i/\partial \delta_j=0$.
    For any $j\in\N_i$, we have
    \begin{align*}
        A_{ij}
        &=\frac{\partial f_i}{\partial \delta_j} = - \frac{\partial P_{g_{ij}}}{\partial \delta_j}g_{ij}^* = \left(\frac{\partial g_{ij}}{\partial \delta_j}g_{ij}^\T+g_{ij}\left(\frac{\partial g_{ij}}{\partial \delta_j}\right)^\T\right)g_{ij}^* \\
        &= \underbrace{\left(g_{ij}^\T g_{ij}^* I_d+g_{ij}{g_{ij}^*}^\T\right)}_{G_{ij}}\frac{\partial g_{ij}}{\partial \delta_j}
        = G_{ij}\frac{P_{g_{ij}}}{\|e_{ij}\|}.
    \end{align*}
    For any $i\in\V$, we have
    \begin{align*}
        A_{ii}
        & = - \sum_{j\in\N_i}\frac{\partial P_{g_{ij}}}{\partial \delta_i}g_{ij}^* = \sum_{j\in\N_i}G_{ij}\frac{\partial g_{ij}}{\partial \delta_i} = -\sum_{j\in\N_i}G_{ij}\frac{P_{g_{ij}}}{\|e_{ij}\|}.
    \end{align*}
    Observe $A_{ii}=-\sum_{j\in\N_i}A_{ij}$ and $A_{ij}=A_{ji}$.
    Therefore, $A$ has a similar structure as graph Laplacian \cite{Mesbahi2010}.

    At the undesired equilibrium $\delta=-2r^*$ where $g_{ij}=-g_{ij}^*$ for all $(i,j)\in\E$, we have
    \begin{align*}
        \left.A_{ij}\right|_{\delta=-2r^*}
        =-\left(I_d+g_{ij}^*{g_{ij}^*}^\T\right)\frac{P_{g_{ij}^*}}{\|e_{ij}\|}
        =- \frac{P_{g^*_{ij}}}{\|e^*_{ij}\|}\le0
    \end{align*}
    for all $j\in\N_i$.
    Similarly, we obtain
    \begin{align*}
        \left.A_{ii}\right|_{\delta=-2r^*}
        &= \sum_{j\in\N_i}\frac{P_{g^*_{ij}}}{\|e^*_{ij}\|}\ge0,\quad\forall i\in\V.
    \end{align*}
    Note $A|_{\delta=-2r^*}$ is positive semi-definite definite.
    To see that, consider any vector $y=[y_1^\T ,\dots,y_n^\T ]^\T$ where $y_i\in\R^d$. Then, $y^\T (A|_{\delta=-2r^*})y=\sum_{(i,j)\in\E} (y_i-y_j)^\T P_{g^*_{ij}}(y_i-y_j)/\|e^*_{ij}\|\ge 0$.
    Thus, $A|_{\delta=-2r^*}$ has at least one positive eigenvalue and consequently the undesired equilibrium $\delta=-2r^*$ is unstable by Lyapunov's indirect method.
\end{proof}

It can be shown that the Jacobian matrix at the desired equilibrium $\delta=0$ is $\left.A\right|_{\delta=0}=-\left.A\right|_{\delta=-2r^*}\le0$, which is symmetric negative semi-definite.
Since $\left.A\right|_{\delta=0}$ is not Hurwitz, the stability of $\delta=0$ cannot be concluded by the Lyapunov's indirect method.
We next present a complete and nonlinear stability analysis of the two equilibriums.

\begin{theorem}[Almost Global Exponential Stability]\label{theorem_global_almostGloablStability}
    Under Assumption \ref{assumption_parallelRigid}, the system trajectory $\delta(t)$ of \eqref{eq_Global_deltaDynamics} exponentially converges to $\delta=0$ from any $\delta(0)\in\S$ except $\delta(0)=-2r^*$.
\end{theorem}
\begin{proof}
Choose the Lyapunov function as
\begin{align*}
    V=\frac{1}{2}\|\delta\|^2.
\end{align*}
The derivative of $V$ is
$\dot{V}=\delta^\T \dot{\delta}=(p-p^*)^\T \dot{p}=-(p^*)^\T \dot{p}$.
Substituting control law \eqref{eq_controlLawMatrix} into $\dot{V}$ yields
\begin{align}\label{eq_VDot_1}
    \dot{V}
    &=-(p^*)^\T \bar{H}^\T \mydiag(P_{g_k})g^*
    =-(e^*)^\T \mydiag(P_{g_k})g^*\nonumber\\
    &=-\sum_{k=1}^m (e_k^*)^\T  P_{g_k} g_k^*
    =-\sum_{k=1}^m \|e_k^*\|(g_k^*)^\T  P_{g_k} g_k^*\le0.
\end{align}
Since $\dot{V}\le0$, we have $\|\delta(t)\|\le \|\delta(0)\|$ for all $t\ge0$.
Furthermore, it follows from Lemma~\ref{lemma_exchangeP1g2} that
\begin{align*}
    (g_k^*)^\T  P_{g_k} g_k^*
    &= g_k^\T  P_{g_k^*} g_k,
\end{align*}
substituting which into \eqref{eq_VDot_1} gives
\begin{align}\label{eq_VDot_2}
    \dot{V}
    &=-\sum_{k=1}^m \|e_k^*\|g_k^\T  P_{g_k^*} g_k =-\sum_{k=1}^m \frac{\|e_k^*\|}{\|e_k\|^2}e_k^\T  P_{g_k^*} e_k \nonumber\\
    &\le-\underbrace{\frac{\min_{k=1,\dots,m}\|e_k^*\|}{4(n-1)s^2}}_{\alpha}\sum_{k=1}^m e_k^\T  P_{g_k^*} e_k,
\end{align}
where the inequality is due to the fact that $\|e_k\|\le2\sqrt{n-1}s$ as given in Corollary~\ref{corollary_maximumpipbarDistance}(b).
Inequality \eqref{eq_VDot_2} can be further written as
\begin{align}\label{eq_VDot_2-2}
    \dot{V}
    &\le-\alpha e^\T \mydiag(P_{g_k^*})e =-\alpha p^\T \bar{H}^\T \mydiag(P_{g_k^*})\bar{H}p \nonumber\\
    &=-\alpha \delta^\T \bar{H}^\T \mydiag(P_{g_k^*})\bar{H}\delta \quad\left(\text{due to $\mydiag(P_{g_k^*})\bar{H}p^*=0$}\right)\nonumber\\
    &=-\alpha \delta^\T \underbrace{\bar{H}^\T \mydiag(P_{g_k^*})}_{\tilde{R}^\T (p^*)}\underbrace{\mydiag(P_{g_k^*})\bar{H}}_{\tilde{R}(p^*)}\delta.
\end{align}
Observe $\tilde{R}(p^*)$ has the same rank and null space as the bearing rigidity matrix $R(p^*)$.
Under the assumption of infinitesimal bearing rigidity, it follows from Theorem~\ref{theorem_conditionInfiParaRigid} that $\Null(\tilde{R}(p^*))=\myspan\{\one\otimes I_d, p^*\}$ and $\rank(\tilde{R}(p^*))=dn-d-1$.
As a result, the smallest $d+1$ eigenvalues of $\tilde{R}^\T (p^*)\tilde{R}(p^*)$ are zero.
Let the minimum positive eigenvalue of $\tilde{R}^\T (p^*)\tilde{R}(p^*)$ be $\lambda_{d+2}$.
Decompose $\delta$ to $\delta=\delta_\perp+\delta_\parallel$, where $\delta_\perp\perp\Null(\tilde{R}(p^*))$ and $\delta_\parallel\in\Null(\tilde{R}(p^*))$.
Then \eqref{eq_VDot_2-2} implies
\begin{align}\label{eq_VDot_3}
    \dot{V}\le -\alpha\lambda_{d+2}\|\delta_\perp\|^2.
\end{align}
Note $\delta_\parallel$ is the orthogonal projection of $\delta$ on $\Null(\tilde{R}(p^*))=\myspan\{\one\otimes I_d, r^*\}$.
Since $\delta\perp \myspan\{\one\otimes I_d\}$, we further know that $\delta_\parallel$ is the orthogonal projection of $\delta$ on $r^*$ (see Figure~\ref{fig_deltaGeometry}).
Let $\theta$ be the angle between $\delta$ and $-r^*$.
Thus, $\|\delta_\perp\|=\|\delta\|\sin\theta$, and \eqref{eq_VDot_3} becomes
\begin{align}\label{eq_VDot_4}
    \dot{V}\le -\alpha\lambda_{d+2}\sin^2\theta\|\delta\|^2.
\end{align}
It can be seen from Figure~\ref{fig_deltaGeometry} that $\theta\in[0,\pi/2)$.
Let $\theta_0$ be the value of $\theta$ at time $t=0$.
Since $\|\delta(t)\|\le \|\delta(0)\|$ for all $t$, it is clear from Figure~\ref{fig_deltaGeometry} that $\theta(t)\ge\theta_0$.
Then, \eqref{eq_VDot_4} becomes
\begin{align*}%\label{eq_VDot_final}
    \dot{V}\le -\underbrace{2\alpha\lambda_{d+2}\sin^2\theta_0}_K V.
\end{align*}
(i)~If $\theta_0>0$, then $K>0$.
As a result, the error $\|\delta(t)\|$ decreases to zero exponentially fast.
(ii)~If $\theta_0=0$, it can be seen from Figure~\ref{fig_deltaGeometry} that $\delta(0)=-2r^*$ which is the undesired equilibrium.
In summary, the system trajectory $\delta(t)$ converges to $\delta=0$ exponentially fast from any initial points except $\delta=-2r^*$.
\end{proof}

In terms of bearings, Theorem~\ref{theorem_global_almostGloablStability} indicates that $g_{ij}(t)$ converges to $g_{ij}^*$ for all $(i,j)\in\E$ from any initial conditions except $g_{ij}(0)=-g_{ij}^*, \forall (i,j)\in\E$.
In addition, the eigenvalue $\lambda_{d+2}$ of $\tilde{R}^\T (p^*)\tilde{R}(p^*)$ affects the convergence rate of the system.
Since $\lambda_{d+2}>0$ if and only if $\G(p^*)$ is infinitesimally bearing rigid, the eigenvalue $\lambda_{d+2}$ can be viewed as a measure of the ``degree of infinitesimal bearing rigidity''.

\subsection{Collision Avoidance}

It is worth noting that there is an implicit assumption in the stability analysis in Theorem~\ref{theorem_global_almostGloablStability} that no two neighbors collide with each other during the formation evolution.
If two neighbors collide, the bearing between them will be mathematically invalid.
As a result, without this assumption, the stability result in Theorem~\ref{theorem_global_almostGloablStability} is merely valid until collision happens.
In fact, control law \eqref{eq_controlLawMatrix} is not able to globally guarantee collision avoidance (see, for example, Figure~\ref{fig_sim_collisionCase}).
In practice, the proposed control law may be implemented together with some other mechanisms like artificial potentials to guarantee collision avoidance.
In this paper, we provide a sufficient condition that ensures all agents maintain a minimum separation distance.% to show that a minimum distance between any agents (even if they are not neighbors) can be ensured if the initial formation is sufficiently close to the target formation.

\begin{figure}
  \centering
  \subfloat[Initial formation]{\includegraphics[width=0.33\linewidth]{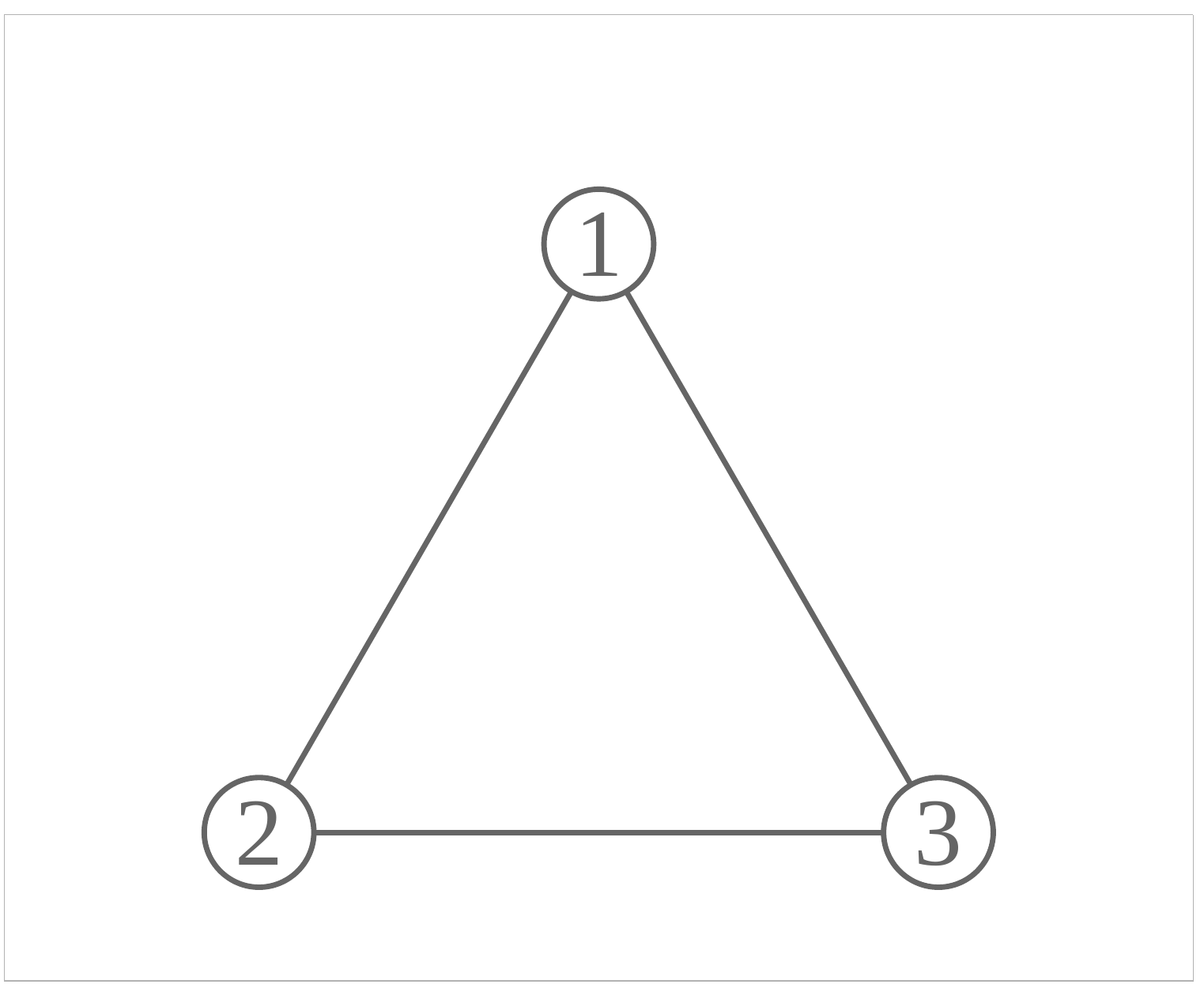}}
  \subfloat[Target formation]{\includegraphics[width=0.33\linewidth]{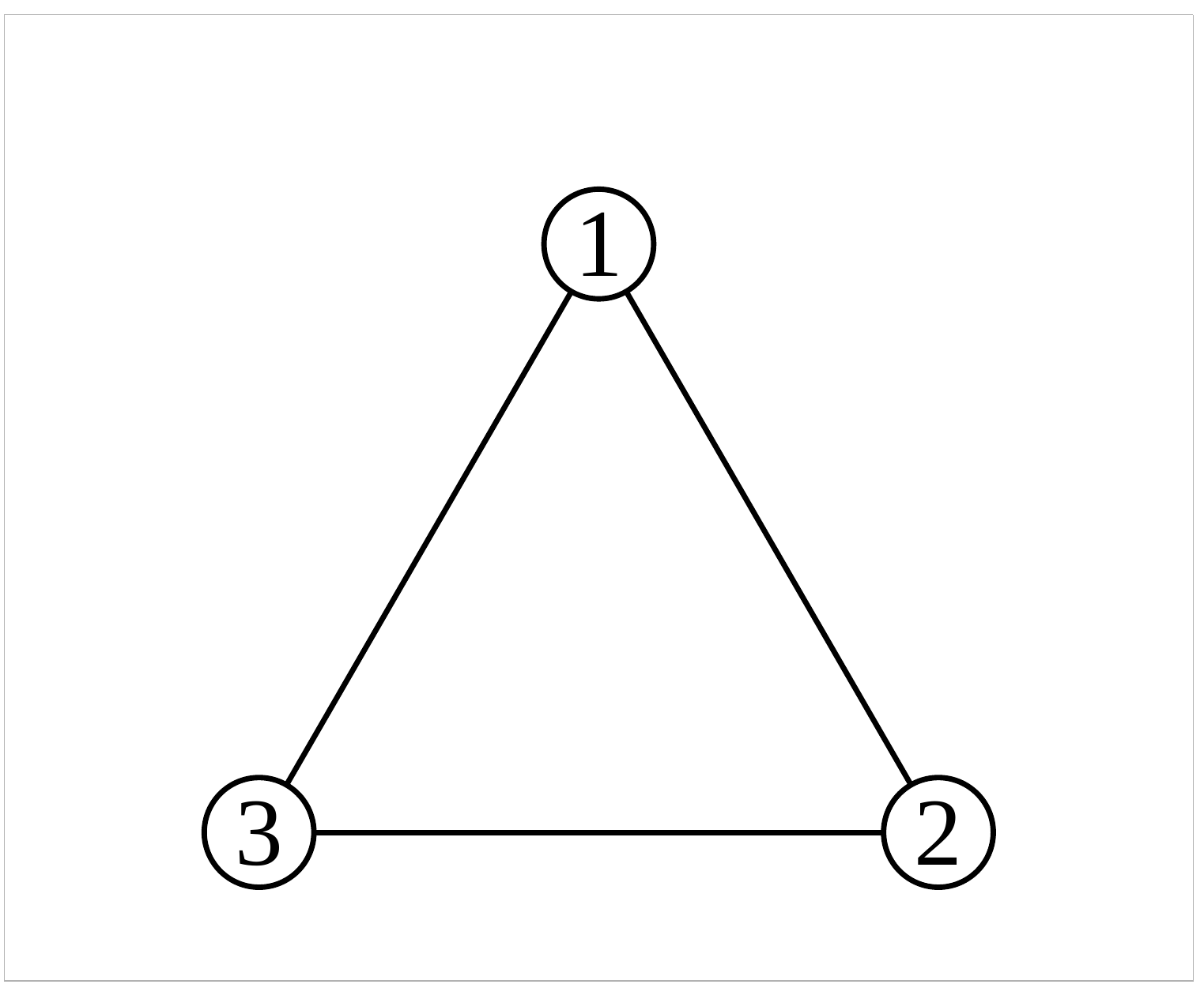}}
  \subfloat[Collision]{\includegraphics[width=0.33\linewidth]{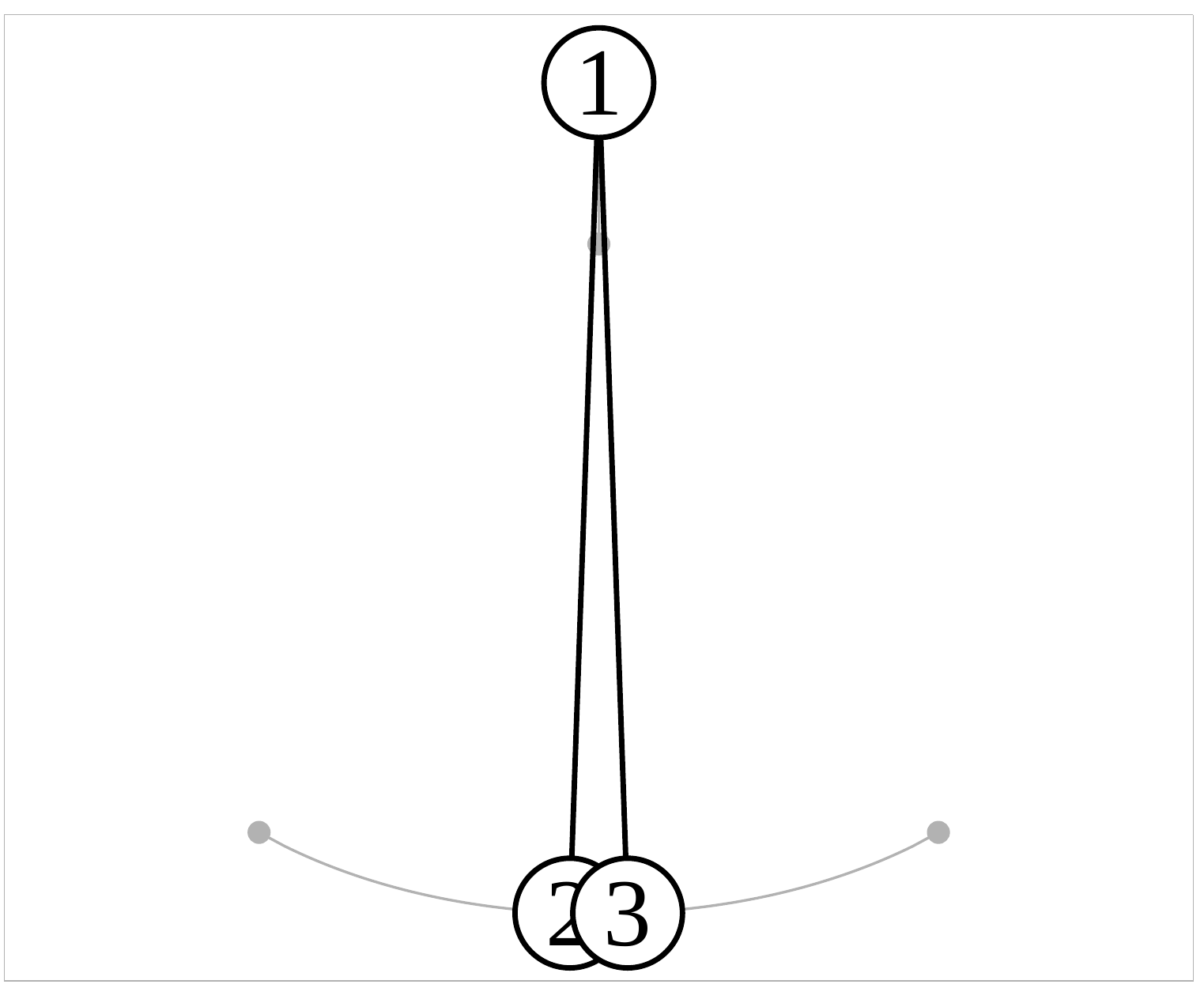}}
  \caption{Control law \eqref{eq_controlLawMatrix} is not able to globally guarantee collision avoidance.
  %As can be seen, with the initial formation given in (a) and the target formation given in (b), agents 2 and 3 will collide as shown in (c).
  }
  \vspace{-10pt}
  \label{fig_sim_collisionCase}
\end{figure}

\begin{theorem}\label{theorem_collisionAvoidanceCondForDelta}
    Under Assumption~\ref{assumption_parallelRigid}, given a minimum distance $\gamma$ satisfying $0\le \gamma<\min_{i,j\in\V}\|p_i^*-p_j^*\|$, it can be guaranteed that $ \|p_i(t)-p_j(t)\|>\gamma, \,\forall i,j\in\V, \forall t\ge0$
%    \begin{align}\label{eq_collisionAvoidanceMinimumDis}
%        \|p_i(t)-p_j(t)\|>\gamma, \quad \forall i,j\in\V, \forall t\ge0
%    \end{align}
    if $\delta(0)$ satisfies
    \begin{align}\label{eq_collisionAvoidanceCondForDelta}
        \|\delta(0)\| < \frac{1}{\sqrt{n}}\left(\min_{i,j\in\V}\|p_i^*-p_j^*\|-\gamma\right).
    \end{align}
\end{theorem}
\begin{proof}
For any $i,j\in\V$ and $t\ge0$, since
\begin{align*}
p_i(t)-p_j(t)\equiv[p_i(t)-p_i^*]-[p_j(t)-p_j^*]+[p_i^*-p_j^*],
\end{align*}
we have
\begin{align*}%\label{eq_interDistanceInequality}
    \|p_i(t)-p_j(t)\|
    &\ge \|p_i^*-p_j^*\|-\|p_i(t)-p_i^*\|-\|p_j(t)-p_j^*\| \nonumber\\
    &\hspace{-65pt}\ge \|p_i^*-p_j^*\|\hspace{-3pt}-\hspace{-3pt}\sum_{\ell=1}^n \|p_\ell(t)-p_\ell^*\| \hspace{-2pt}\ge \hspace{-2pt}\|p_i^*-p_j^*\|\hspace{-3pt}-\hspace{-3pt}\sqrt{n}\|p(t)-p^*\|.\nonumber
  %  &\ge\|p_i^*-p_j^*\|-\sqrt{n}\|p(t)-p^*\|.
\end{align*}
Substituting $\delta(t)=p(t)-p^*$ and $\|\delta(t)\|\le\|\delta(0)\|$ into the above inequality gives
\begin{align*}
    \|p_i(t)-p_j(t)\|
    \ge\|p_i^*-p_j^*\|-\sqrt{n}\|\delta(0)\|.
\end{align*}
As a result, if \eqref{eq_collisionAvoidanceCondForDelta} holds, we have the desired result.
\end{proof}

The upper bound for $\|\delta(0)\|$ given in Theorem~\ref{theorem_collisionAvoidanceCondForDelta} is inversely proportional to $\sqrt{n}$.
This is intuitively reasonable since the chance for two agents colliding is high when the number of the agents is large and consequently, the initial error must be small to avoid collision.
In addition, the condition given in Theorem~\ref{theorem_collisionAvoidanceCondForDelta} is conservative.
Extensive simulations have shown that the proposed controller can avoid collisions even if the above condition is not satisfied.

\section{Bearing-Only Formation Control without a Global Reference Frame}\label{section_BOF_noGlobal}

In this section, we study the case where the global reference frame is unknown to the agents and each agent can only measure the bearings and relative orientations of their neighbors in their local reference frames.

Consider $n\ge2$ agents in $\R^3$.
Denote $p_i\in\R^3$, $v_i\in\R^3$, and $w_i\in\R^3$ as the position, linear velocity, and angular velocity of agent $i\in\V$ expressed in a global reference frame which is unknown to each agent.
There is a local reference frame fixed on the body of each agent.
We use the superscript $b$ to indicate a vector expressed in the local body frame.
A vector quantity without the superscript is expressed in the global frame.
In particular, $v_i^b$ and $w_i^b$ represent the linear velocity and angular velocity of agent $i$ expressed in its own body frame.
Let $Q_i\in SO(3)$ be the rotation form the body frame of agent $i$ to the global frame.
Then, $v_i=Q_i v_i^b$ and $w_i=Q_i w_i^b$.
The position and orientation dynamics of agent $i$ is
\begin{align}\label{eq_noGlobal_agentDynamics}
    \dot{p}_i &=Q_iv_i^b, \nonumber\\
    \dot{Q}_i &= Q_i\sk{w_i^b},
\end{align}
where $\sk{\,\cdot\,}$ is the skew-symmetric matrix operator defined in \eqref{eq_skewSymmetricOperator}, and $v_i^b$ and $w_i^b$ are the inputs to be designed.

Denote, as before, $e_{ij}\triangleq p_j-p_i$ and $g_{ij}\triangleq {e_{ij}}/{\|e_{ij}\|}$ for $(i,j)\in\E$.
Agent $i$ can measure the bearings of its neighbors in its local frame, $\{g_{ij}^b\}_{j\in\N_i}$, where $g_{ij}^b=Q_i^\T g_{ij}$.
Moreover, assume agent $i$ can also measure the \emph{relative orientation} of its neighbors, $\{Q_i^\T Q_j\}_{j\in\N_i}$.
The bearing-only formation control problem to be solved in this section is stated as below.

\begin{problem}\label{problem_noGlobal_bearingonlyformationcontrol}
    Given feasible constant bearing constraints $\{g_{ij}^*\}_{(i,j)\in\E}$ and an initial formation $\G(p(0))$ with agent orientations as $\{Q_i(0)\}_{i\in\V}$, design $v_i^b(t)$ and $w_i^b(t)$ for agent $i\in\V$ based only on the local bearing measurements $\{g_{ij}^b(t)\}_{j\in\N_i}$ and relative orientation measurements $\{Q_i^\T(t) Q_j(t)\}_{j\in\N_i}$ such that $\{Q_i(t)\}_{i\in\V}$ converge to a common value and $g_{ij}^b(t)\rightarrow g_{ij}^*$ as $t\rightarrow\infty$ for all $(i,j)\in\E$.
\end{problem}

It is notable that there is an orientation synchronization problem embedded in Problem~\ref{problem_noGlobal_bearingonlyformationcontrol}.
This scheme is inspired by the works on formation control based on orientation alignment \cite{Oh2014ICCAS,KKOh2014TAC}.
Once the orientations of the agents have synchronized, the synchronized local frames can be viewed as a common frame where the bearing constraints should be satisfied.
It is worth mentioning that the value of the finally synchronized orientation is not of our interest, and we only care about the shape of the formation.
If required in practice, one may introduce a leader to control the synchronized orientation.

\vspace{-10pt}
\subsection{A Bearing-Only Control Law}

The proposed position and orientation control laws are
\begin{subequations}\label{eq_controlLaw_noGlobalReferece}
\begin{align}
    \label{eq_controlLaw_noGlobalReferece_position}
    v_i^b
    &= -\sum_{j\in\mathcal{N}_i} P_{g^b_{ij}}(I_3+Q_i^\T Q_j) g_{ij}^*, \\
    \label{eq_controlLaw_noGlobalReferece_orientation}
    \sk{w_i^b}
    &= -\sum_{j\in\mathcal{N}_i} \left(Q_j^\T Q_i-Q_i^\T Q_j\right).
\end{align}
\end{subequations}
The proposed control law is distributed and can be implemented without the knowledge of the global frame.
It only requires local bearing measurements $\{g_{ij}^b\}_{j\in\N_i}$ and relative orientation measurements $\{Q_i^\T Q_j\}_{j\in\N_i}$.
Control law \eqref{eq_controlLaw_noGlobalReferece_orientation} actually is the orientation synchronization control proposed in \cite{Igarashi2009TCST}.
Substituting control law \eqref{eq_controlLaw_noGlobalReferece} into \eqref{eq_noGlobal_agentDynamics} gives the closed-loop system dynamics with all vector quantities expressed in the global frame as
\begin{subequations}\label{eq_noGlobal_closedLoopDynamics}
\begin{align}
    \label{eq_noGlobal_closedLoopDynamics_position}
    \dot{p}_i
    &= -\sum_{j\in\mathcal{N}_i} P_{g_{ij}}(Q_i+Q_j) g_{ij}^*, \\
    \label{eq_noGlobal_closedLoopDynamics_orientation}
    \dot{Q}_i
    &= -\sum_{j\in\mathcal{N}_i} Q_i\left(Q_j^\T Q_i-Q_i^\T Q_j\right).
\end{align}
\end{subequations}
While deriving \eqref{eq_noGlobal_closedLoopDynamics_position}, we use the fact that $g_{ij}=Q_ig_{ij}^b$ and $Q_iP_{g_{ij}^b}Q_i^\T =P_{g_{ij}}$.

We next show that the centroid and the scale of the formation are invariant under control law \eqref{eq_controlLaw_noGlobalReferece}.

\begin{lemma}\label{lemma_noGlobal_xDotPerpendicular}
    Under control law \eqref{eq_controlLaw_noGlobalReferece}, $\dot{p}\perp \myspan\left\{\one\otimes I_3, p\right\}.$
%    \begin{align}\label{eq_noGlobal_pdot_perpendicular}
%        \dot{p}\perp \myspan\left\{\one\otimes I_3, p\right\}.
%    \end{align}
\end{lemma}
\begin{proof}
    Let $Q_{ij}\triangleq Q_i+Q_j$.
    Then, $\dot{p}_i=-\sum_{j\in\mathcal{N}_i} P_{g_{ij}}Q_{ij}g_{ij}^*$.
    Consider an arbitrary oriented graph, the position dynamics \eqref{eq_noGlobal_closedLoopDynamics_position} can be written in a matrix form as $\dot{p} = \bar{H}^\T\dia{P_{g_k}}\dia{Q_{k}}g^*$.
    Because $\one\otimes I_3$ and $p$ are all in the left null space of $\bar{H}^\T\dia{P_{g_k}}$, we obtain the result.%\eqref{eq_noGlobal_pdot_perpendicular}.
\end{proof}

\begin{theorem}[Centroid and Scale Invariance]\label{theorem_noGlobal_invariantCentroidScale}
    The centroid $\bar{p}$ and the scale $s$ are invariant under control law \eqref{eq_controlLaw_noGlobalReferece}.
\end{theorem}
\begin{proof}
     With Lemma~\ref{lemma_xDotPerpendicular}, the proof is similar to Theorem~\ref{propos_invariantCentroidScale}.
\end{proof}
\begin{remark}
    It can be easily verified that Lemma~\ref{lemma_noGlobal_xDotPerpendicular} and Theorem~\ref{theorem_noGlobal_invariantCentroidScale} hold for any position control law that has the form of $\dot{p}_i=-\sum_{i=1}^n P_{g_{ij}}y_{ij}$ where $y_{ij}\in\R^3$ and $y_{ij}=-y_{ji}$.
\end{remark}

The following results, which can be obtained from Theorem~\ref{theorem_noGlobal_invariantCentroidScale}, give bounds for $\max_{i\in\V}\|p_i(t)-\bar{p}\|$ and $\|p_i(t)-p_j(t)\|, \forall i,j\in\V$.

\begin{corollary}
The formation trajectory under the control law \eqref{eq_controlLaw_noGlobalReferece} satisfies the following inequalities,
\begin{enumerate}[(a)]
	\item $s\le \max_{i\in\V}\|p_i(t)-\bar{p}\| \le s\sqrt{n-1}, \quad \forall t\ge0.$
	\item $ \|p_i(t)-p_j(t)\|\le 2s\sqrt{n-1},\quad\forall i,j\in\V,\ \forall t\ge0.$
\end{enumerate}
\end{corollary}
\begin{proof}
The proof is similar to Corollary~\ref{corollary_maximumpipbarDistance}.
\end{proof}

\vspace{-10pt}
\subsection{Formation Stability Analysis}

The closed-loop system \eqref{eq_noGlobal_closedLoopDynamics} is a cascade system: the dynamics of the orientation is independent to the dynamics of the position, whereas the converse is not true.
Similar cascade systems can also be found in recent studies on formation control in $SE(2)$ or $SE(3)$ \cite{Oh2012CDC,Oh2014ICCAS,KKOh2014TAC,Montijano2014ACC} and input-to-state stability (ISS) can be used to prove the stability of the cascade systems.
In order to analyze the stability of system \eqref{eq_noGlobal_closedLoopDynamics}, we first note that the orientation of the agents will synchronize by control law \eqref{eq_controlLaw_noGlobalReferece_orientation} under the following assumption \cite{Igarashi2009TCST}.

\begin{assumption}\label{assumption_initialAttitudePositiveDefinite}
    In the initial formation, there exists $Q_0\in SO(3)$ such that $Q_0^\T Q_i$ is (non-symmetric) positive definite for all $i\in\V$.
\end{assumption}

\begin{remark}
Based on axis-angle representation, a rotation matrix is positive definite if and only if the rotation angle is in $(-\pi/2,\pi/2)$.
\end{remark}

\begin{lemma}[{\cite[Thm~1]{Igarashi2009TCST}}]
    \label{lemma_attitudeSynchronization}
    Under Assumption~\ref{assumption_initialAttitudePositiveDefinite}, if the interconnection graph is fixed and strongly connected, the orientation control law \eqref{eq_controlLaw_noGlobalReferece_orientation} guarantees orientation synchronization in the sense that $\lim_{t\rightarrow \infty} Q_i^\T Q_j=I_3$ for all $i,j\in\V$.
\end{lemma}
Although the value of the final converged orientation is not given in \cite{Igarashi2009TCST}, there exists a unique $Q^*\in SO(3)$ such that $Q_i$ ($i\in\V$) converges to $Q^*$ asymptotically.
The specific value of $Q^*$ is not of our interest and it is not required to prove the formation stability.
In fact, control law \eqref{eq_controlLaw_noGlobalReferece_orientation} can be replaced by any other orientation control law as long as it ensures the orientations can converge to a common constant value.
With the above preparation, we next define the target formation that the formation will converge to.

\begin{definition}[Target Formation]\label{definition_noGlobalcase_targetFormation}
    Let $\G(p^*)$ be the \emph{target formation} that satisfies
    \begin{enumerate}[(a)]
        \item Centroid: $\bar{p}^*=\bar{p}(0)$.
        \item Scale: $s^*=s(0)$.
        \item Bearing: $(p_j^*-p_i^*)/\|p_j^*-p_i^*\|=Q^*g_{ij}^*$ for all $(i,j)\in\E$.
    \end{enumerate}
\end{definition}

\begin{lemma}[Existence and Uniqueness]%\label{lemma_noGlobalcase_targetFormationExistUnique}
The target formation $\G(p^*)$ in Definition~\ref{definition_noGlobalcase_targetFormation} always exists and is unique under Assumptions~\ref{assumption_parallelRigid} and \ref{assumption_initialAttitudePositiveDefinite}.
\end{lemma}
\begin{proof}
The proof is similar to Lemma~\ref{lemma_targetFormationExistUnique}.
But it should be noted that the bearings of $\G(p^*)$ in Definition~\ref{definition_noGlobalcase_targetFormation} are $\{Q^*g_{ij}^*\}_{(i,j)\in\E}$ instead of $\{g_{ij}^*\}_{(i,j)\in\E}$.
\end{proof}

Let $\delta_i\triangleq p_i-p^*_i$.
It follows from the closed-loop position dynamics \eqref{eq_noGlobal_closedLoopDynamics_position} that
\begin{align*}
    \dot{\delta}_i
    &=-\sum_{j\in\mathcal{N}_i} P_{g_{ij}}(Q_i+Q_j)g_{ij}^* \nonumber\\
    &=\underbrace{-2\sum_{j\in\mathcal{N}_i} P_{g_{ij}}{Q^*}g_{ij}^*}_{f_i(\delta)}
    + \underbrace{\sum_{j\in\mathcal{N}_i} P_{g_{ij}}(2Q^*-Q_i-Q_j)g_{ij}^*}_{h_i(t)}.
\end{align*}
Denote $\delta=[\delta_1^\T ,\dots,\delta_n^\T ]^\T$, $f(\delta)=[f_1^\T (\delta),\dots,f_n^\T (\delta)]^\T$, and $h(t)=[h_1^\T (t),\dots,h_n^\T (t)]^\T$.
Then, the $\delta$-dynamics is
\begin{align}\label{eq_noGlobal_deltaDynamics}
    \dot{\delta} = f(\delta) + h(t),
\end{align}
where $h(t)$ can be viewed as an input.
It should be noted that the autonomous system (i.e., system \eqref{eq_noGlobal_deltaDynamics} with $h(t)\equiv0$)
\begin{align*}
  \dot{\delta}=f(\delta)
\end{align*}
has already been well studied in Section~\ref{section_BOF_global}.
For this autonomous system, we know from Section~\ref{section_BOF_global} that $\delta=0$ is an almost globally stable equilibrium and $g_{ij}(t)\rightarrow Q^*g_{ij}^*$ almost globally as $t\rightarrow\infty$.

\begin{lemma}\label{lemma_noGlobal_inputConvergeZero}
    The input $h(t)$ converges to zero asymptotically.
\end{lemma}
\begin{proof}
Note
$\|h(t)\|
    \le \sum_{i=1}^n \|h_i(t)\|
    \le \sum_{i=1}^n\sum_{j\in\mathcal{N}_i} \|P_{g_{ij}}\|\|2Q^*-Q_i-Q_j\|\|g_{ij}^*\|$.
Since $Q_i, Q_j\rightarrow Q^*$ by Lemma~\ref{lemma_attitudeSynchronization} and $\|P_{g_{ij}}\|=\|g_{ij}^*\|=1$, we have $\|h(t)\|\rightarrow0$ as $t\rightarrow\infty$.
\end{proof}

We next identify the state manifold and the equilibriums of the $\delta$-dynamics \eqref{eq_noGlobal_deltaDynamics}.
Denote, as before, $r(t)=p(t)-\one\otimes\bar{p}$ and $r^*=p^*-\one\otimes\bar{p}^*$.

\begin{lemma}
    System \eqref{eq_noGlobal_deltaDynamics} evolves on the surface of the sphere
    \begin{align*}
        \S=\{\delta\in\R^{3n}: \|\delta+r^*\|=\|r^*\|\}.
    \end{align*}
\end{lemma}

\begin{proof}
It follows from $\delta(t)=r(t)-r^*$ that $\|\delta(t)+r^*\|=\|r(t)\|=\|r^*\|$, where $\|r(t)\|=\|r^*\|$ is due to the scale invariance.
\end{proof}

\begin{theorem}[Equilibrium]
    Under Assumptions~\ref{assumption_parallelRigid} and \ref{assumption_initialAttitudePositiveDefinite},
    the closed-loop system \eqref{eq_noGlobal_closedLoopDynamics} (i.e., the $\delta$-dynamics together with the orientation dynamics) has two equilibrium points,
    \begin{enumerate}[(a)]
        \item $\delta=0$ and $Q_i=Q^*, \forall i\in\V$,
        \item $\delta=-2r^*$ and $Q_i=Q^*, \forall i\in\V$.
    \end{enumerate}
\end{theorem}

\begin{proof}
Any equilibrium must satisfy
\begin{align}\label{eq_noGlobal_equilibriumCondition}
    \sum_{j\in\mathcal{N}_i} P_{g_{ij}}(Q_i+Q_j) g_{ij}^*=0, \quad\forall i\in\V.
\end{align}
It follows from Lemma~\ref{lemma_attitudeSynchronization} that $Q_i=Q^*$ ($\forall i\in\V$) is the equilibrium for the orientation dynamics \eqref{eq_noGlobal_closedLoopDynamics_orientation} under Assumption~\ref{assumption_initialAttitudePositiveDefinite}.
Then, \eqref{eq_noGlobal_equilibriumCondition} becomes
\begin{align*}
    \sum_{j\in\mathcal{N}_i} P_{g_{ij}}Q^* g_{ij}^*=0, \quad\forall i\in\V.
\end{align*}
Similar to the proof of Theorem~\ref{theorem_global_twoEquilibrium}, it can be shown that the above equation suggests two equilibriums: $\delta=0$ and $\delta=-2r^*$.
The bearings at the two equilibriums are $g_{ij}=Q^*g_{ij}^*, \forall(i,j)\in\E$ and $g_{ij}=-Q^*g_{ij}^*, \forall(i,j)\in\E$, respectively.
\end{proof}

The equilibrium $\delta=0$ is desired while the other one $\delta=-2r^*$ is undesired.
The formations at the two equilibriums have the same centroid, scale, and shape, but they have the opposite bearings.
We next present the main stability result and show that the desired equilibrium $\delta =0$ is almost globally stable.
Since the equilibrium $\delta=0$ for $\dot{\delta}=f(\delta)$ is almost globally stable, the idea of the proof is to show system \eqref{eq_noGlobal_deltaDynamics} is almost globally ISS \cite{Angeli2011TAC} and then the almost global stability can be concluded by $\lim_{t\rightarrow\infty}h(t)=0$.
Note the conventional ISS is not applicable since it is defined for globally stable equilibriums.

\begin{theorem}[Almost Global Asymptotical Stability]\label{theorem_noGlobal_almostGloablStability}
    Under Assumptions~\ref{assumption_parallelRigid} and \ref{assumption_initialAttitudePositiveDefinite}, the system trajectory $\delta(t)$ of \eqref{eq_noGlobal_deltaDynamics} asymptotically converges to $\delta=0$ from any $\delta(0)\in\S$ except a set of measure zero.
\end{theorem}

\begin{proof}
We first prove system \eqref{eq_noGlobal_deltaDynamics} fulfills the ultimate boundedness property \cite[Proposition~3]{Angeli2011TAC}.
Consider the Lyapunov function $V=\|\delta\|^2/2$.
For the autonomous system $\dot{\delta}=f(\delta)$, we already know from the proof of Theorem~\ref{theorem_global_almostGloablStability} that there exists a positive constant $\kappa$ such that
\begin{align*}
 \frac{\partial V}{\partial \delta}f(\delta)\le-\kappa\sin^2\theta\|\delta\|^2=-\kappa\left(1-\frac{\|\delta\|^2}{4\|r^*\|^2}\right)\|\delta\|^2.
\end{align*}
The derivative of $V$ along the trajectory of system \eqref{eq_noGlobal_deltaDynamics} is
\begin{align*}%\label{eq_noGlobal_VDot1}
    \dot{V}\hspace{-3pt}
    &=\hspace{-3pt}\frac{\partial V}{\partial \delta}(f(\delta)+h(t)) \hspace{-3pt}\le \hspace{-3pt}-\kappa\left(1-\frac{\|\delta\|^2}{4\|r^*\|^2}\right)\|\delta\|^2+\|\delta\|\|h(t)\| \nonumber\\
    &=-\kappa\|\delta\|^2+\frac{\kappa\|\delta\|^4}{4\|r^*\|^2}+\|\delta\|\|h(t)\| \nonumber\\
    &\le-2\kappa V+4\kappa\|r^*\|^2+2\|r^*\|\|h(t)\|,
\end{align*}
where the last inequality is due to $\|\delta\|\le2\|r^*\|$.
By \cite[Proposition~3]{Angeli2011TAC}, system \eqref{eq_noGlobal_deltaDynamics} fulfills the ultimate boundedness property.

We next show system \eqref{eq_noGlobal_deltaDynamics} satisfies the three assumptions A0-A2 in \cite{Angeli2011TAC}.
First, the state of \eqref{eq_noGlobal_deltaDynamics} evolves on the sphere $\S$ which satisfies assumption A0.
Second, consider $V=\|\delta\|^2/2$.
For the autonomous system $\dot{\delta}=f(\delta)$, we have $(\partial V/\partial \delta) f(\delta)\le -\kappa\sin^2\theta\|\delta\|^2<0$ for all $\delta\in\S$ except the equilibriums $\delta=0$ and $\delta=-2r^*$.
Thus, assumption A1 is fulfilled.
Third, the unstable equilibrium of the autonomous system $\dot{\delta}=f(\delta)$ is $\delta=-2r^*$.
It is isolated.
Similar to the proof of Proposition~\ref{proposition_global_Jacobian}, it can be shown that the Jacobian $A=\partial f/\partial \delta$ at $\delta=-2r^*$ is positive semi-definite and at least one eigenvalue is positive.
As a result, assumption A2 is fulfilled.

Thus, it can be concluded from \cite[Proposition~2]{Angeli2011TAC} that system \eqref{eq_noGlobal_deltaDynamics} is almost globally ISS.
Furthermore, since the input $h(t)$ converges to zero as shown in Lemma~\ref{lemma_noGlobal_inputConvergeZero}, the equilibrium $\delta=0$ is almost globally asymptotically stable. The trajectory of \eqref{eq_noGlobal_deltaDynamics} asymptotically converges to $\delta=0$ from any $x(0)\in\S$ except a set of zero measure.
\end{proof}
\begin{remark}
    In terms of bearings, Theorem~\ref{theorem_noGlobal_almostGloablStability} indicates that $g_{ij}(t)$ almost globally converges to $Q^*g_{ij}^*$ for all $(i,j)\in\E$.
    Consequently, $g_{ij}^i(t)=Q_i^\T(t) g_{ij}(t)\rightarrow (Q^*)^\T Q^*g_{ij}^*=g_{ij}^*$ as $t\rightarrow\infty$.
    Therefore, control law \eqref{eq_controlLaw_noGlobalReferece} solves Problem~\ref{problem_noGlobal_bearingonlyformationcontrol}.
\end{remark}

\section{Simulation Examples}\label{section_simulation}

%\subsection{Formation Control with a Global Reference Frame}

In order to illustrate control law \eqref{eq_controlLawElement}, we have already presented two simulation examples in Figure~\ref{fig_simExampleDemo}.
It is worth noting that collinear initial formations may cause troubles for distance-based formation control, but as shown in Figure~\ref{fig_simExampleDemo}(b) it is not a problem for bearing-only formation control.
Two more simulation examples are shown in Figures~\ref{fig_sim_Global_2DPolygon} and \ref{fig_sim_Global_3DCube}, respectively.
The initial formations are generated randomly.
It is shown that control law \eqref{eq_controlLawMatrix} can steer the agents to a formation that satisfies the bearing constraints.

%\subsection{Formation Control without a Global Reference Frame}

In order to illustrate control law \eqref{eq_controlLaw_noGlobalReferece}, two simulation examples are shown in Figures~\ref{fig_sim_noGlobal_2DSquare} and \ref{fig_sim_noGlobal_3DCube}, respectively.
The local frame for each agent is represented by the line segments in red/solid, green/dashed, and blue/dotted in the figures.
The initial positions and orientations of the agents are generated randomly.
The target formations in Figures~\ref{fig_sim_noGlobal_2DSquare} and \ref{fig_sim_noGlobal_3DCube} have the same shape as those in Figures~\ref{fig_simExampleDemo}(b) and \ref{fig_sim_Global_3DCube}, respectively.
As can be seen, the orientations of the agents finally synchronize, and the bearing constraints are satisfied in the synchronized frames.

\section{Conclusions and Future Works}\label{section_conclusion}

In this paper, we first proposed a bearing rigidity theory that is applicable to arbitrary dimensions and showed that the shape of a framework can be uniquely determined by its inter-neighbor bearings if and only if it is infinitesimally bearing rigid.
The infinitesimal bearing rigidity of a given framework can be conveniently examined by a rank condition.
The connection between the proposed bearing rigidity and the well-known distance rigidity has also been explored.
We showed that a framework in $\R^2$ is infinitesimally bearing rigid if and only if it is also infinitesimally distance rigid.
Based on the bearing rigidity theory, we studied the problem of bearing-only stabilization of multi-agent formations.
Two bearing-only distributed formation control laws have been proposed, respectively, for the cases with and without global reference frames.
Almost global formation stability for the control laws has been proved.

It is assumed in this paper that the underlying graphs for the frameworks or formations are undirected.
In fact, the bearing rigidity theory is independent to whether the underlying graph is undirected or directed.
By considering a directed graph (i.e., an orientation of the undirected graph), the bearing rigidity results are still valid.
The fundamental reason is that one of the two bearings $g_{ij}$ and $g_{ji}$ is redundant since $g_{ij}=-g_{ji}$.
However, the stability analysis of the bearing-only formation control laws is merely valid for the undirected case.
It is meaningful to study bearing-only formation control with directed interaction topologies in the future.

\begin{figure}
  \centering
  \subfloat[Initial formation]{\includegraphics[width=0.35\linewidth]{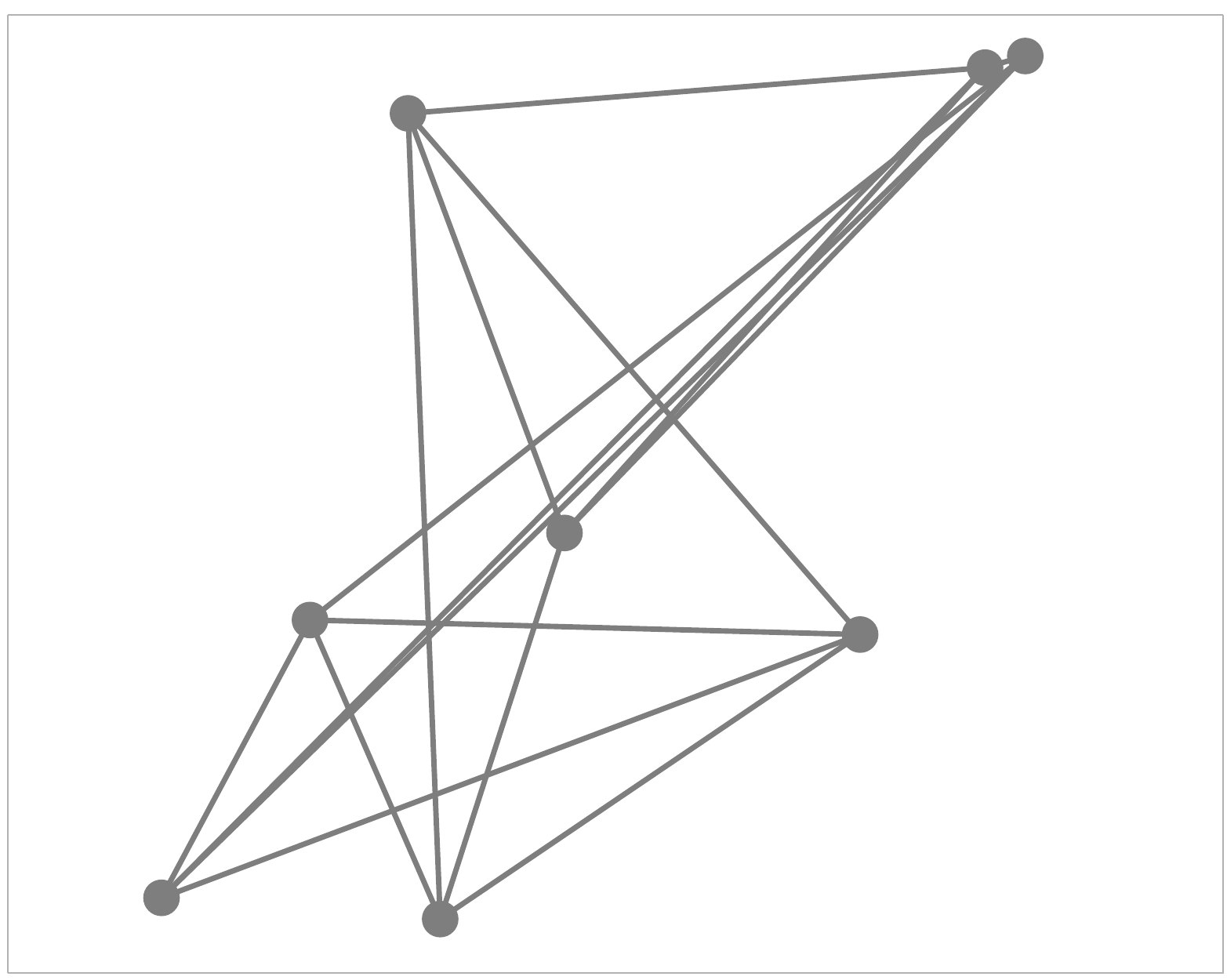}}
  \subfloat[Final formation]{\includegraphics[width=0.35\linewidth]{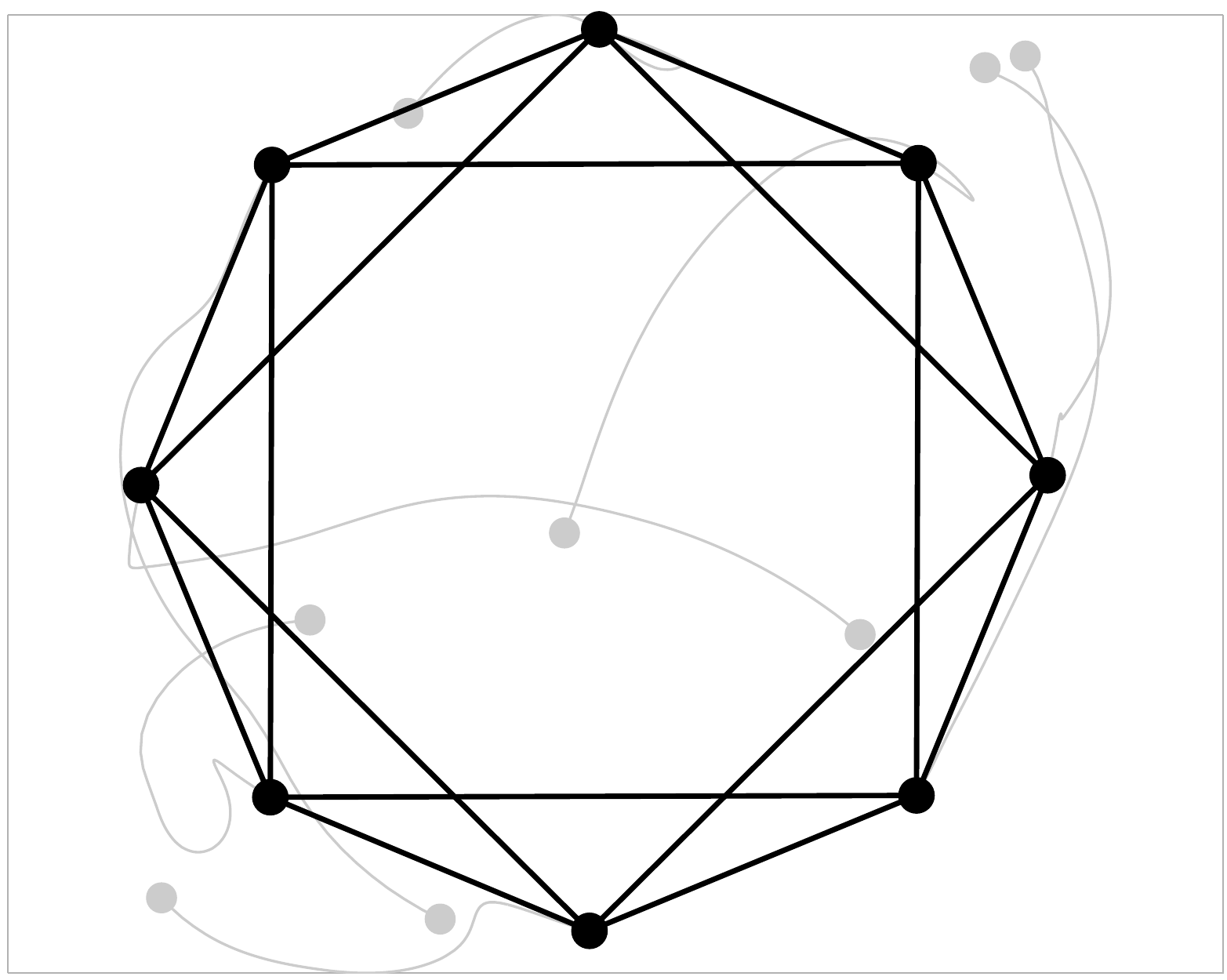}}
    %\vspace{-5pt}
  \caption{The case with a global reference frame in $\R^2$ with $n=8$, $m=16$.}
  \label{fig_sim_Global_2DPolygon}
  \vspace{-10pt}
\end{figure}
\begin{figure}
  \centering
  \subfloat[Initial formation]{\includegraphics[width=0.35\linewidth]{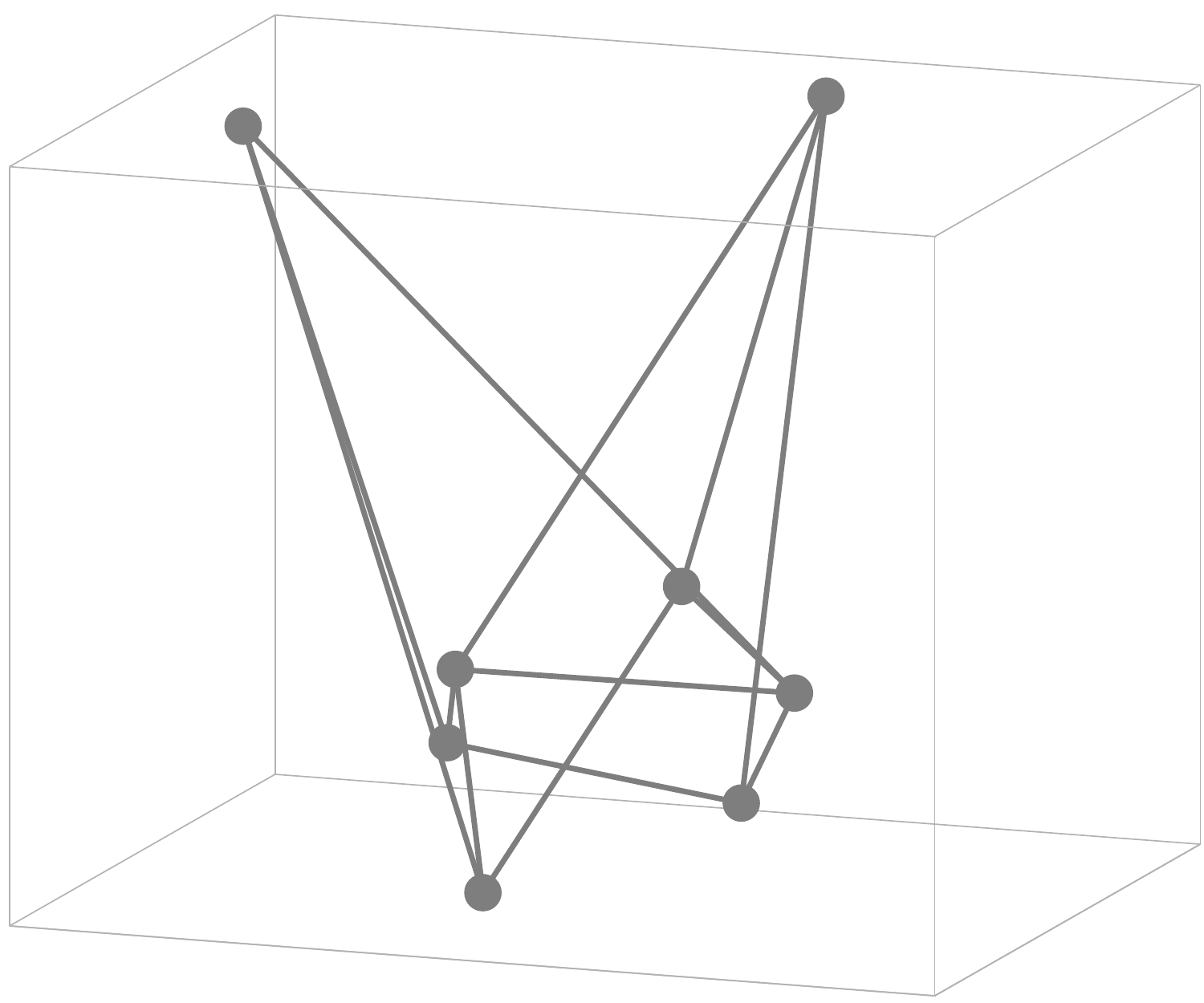}}
  \subfloat[Final formation]{\includegraphics[width=0.35\linewidth]{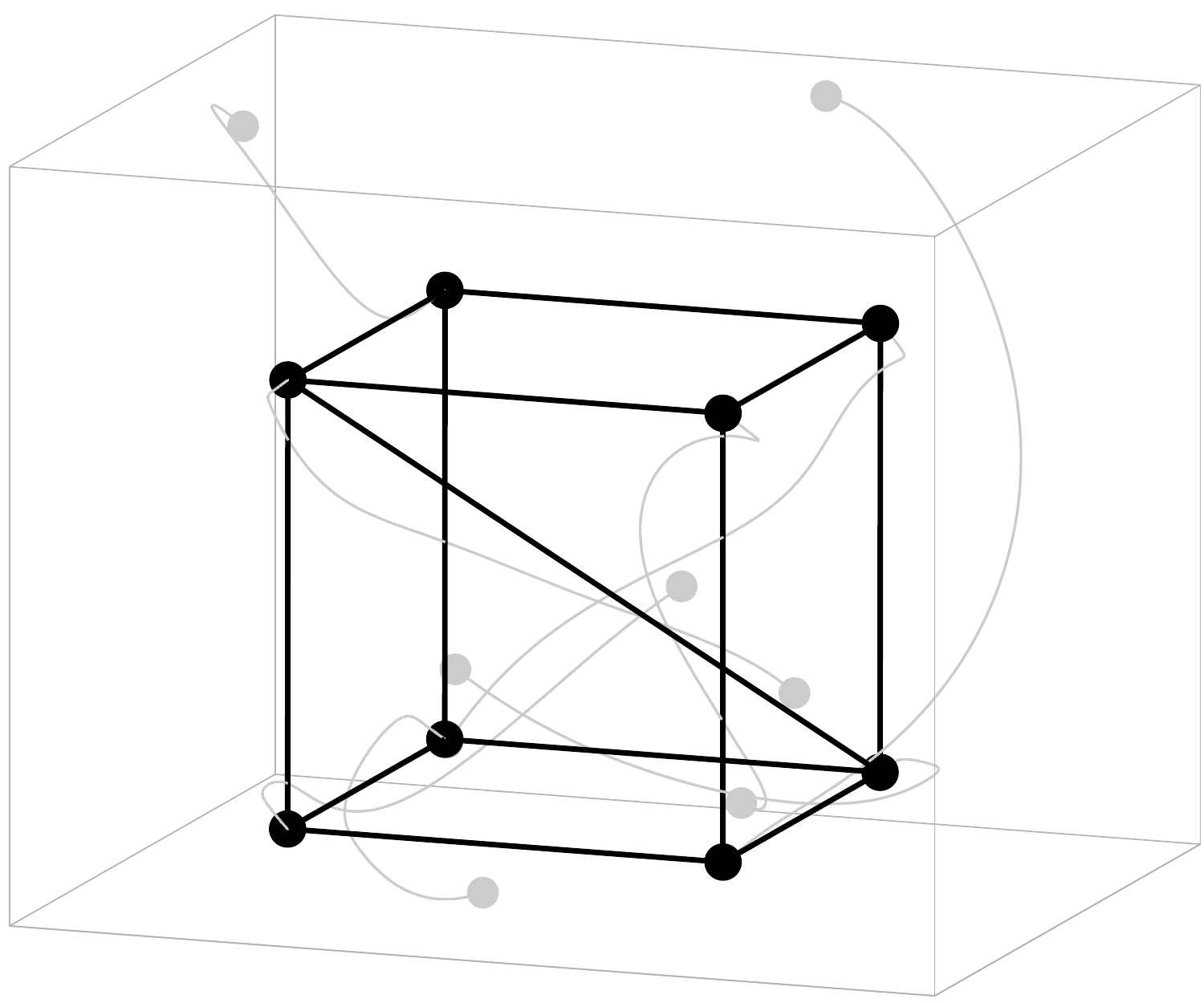}}
    %\vspace{-5pt}
  \caption{The case with a global reference frame in $\R^3$ with $n=8$, $m=13$.}
  \label{fig_sim_Global_3DCube}
    \vspace{-10pt}
\end{figure}
\begin{figure}[t]
  \centering
  \subfloat[Initial formation]{\includegraphics[width=0.35\linewidth]{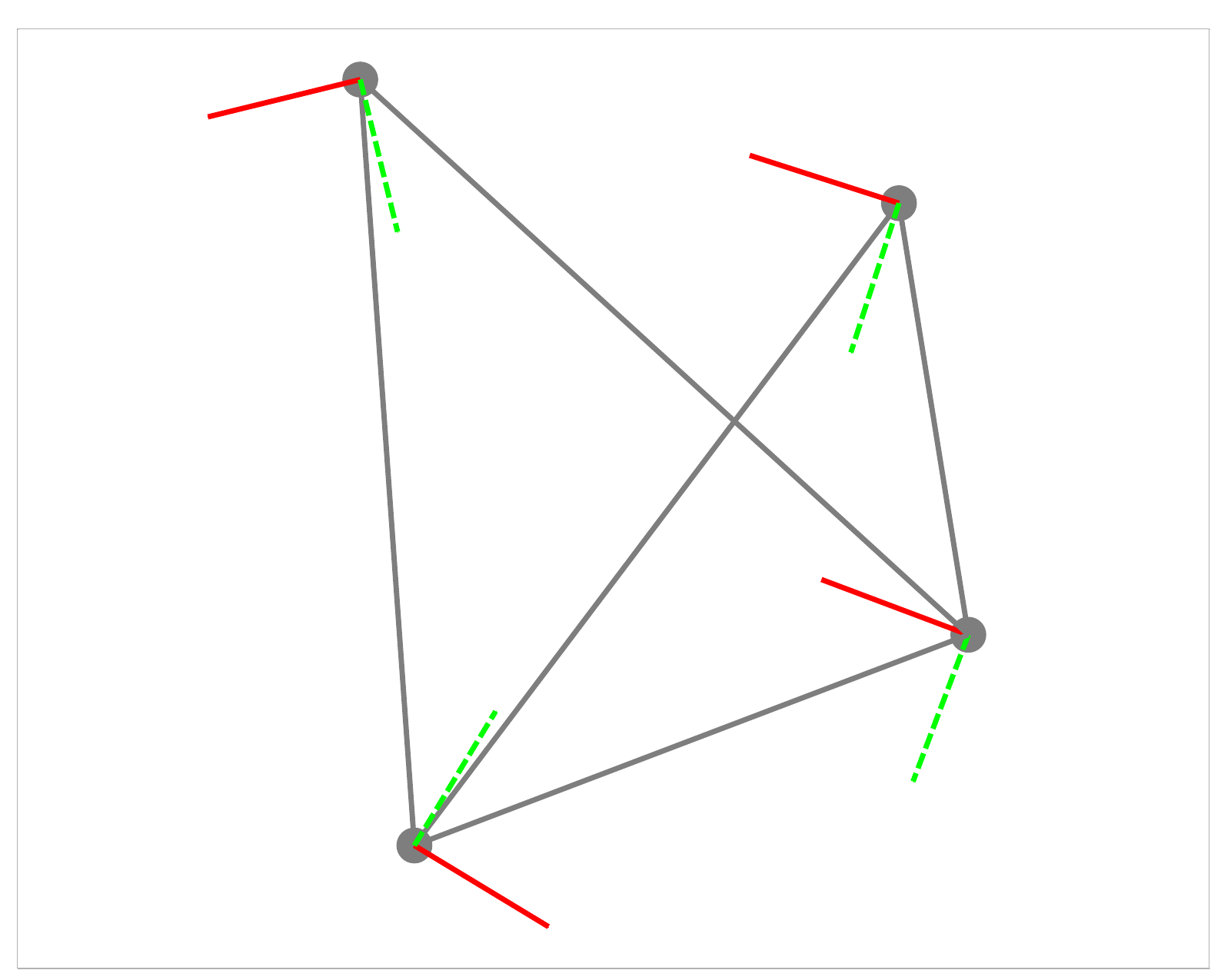}}
  \subfloat[Final formation]{\includegraphics[width=0.35\linewidth]{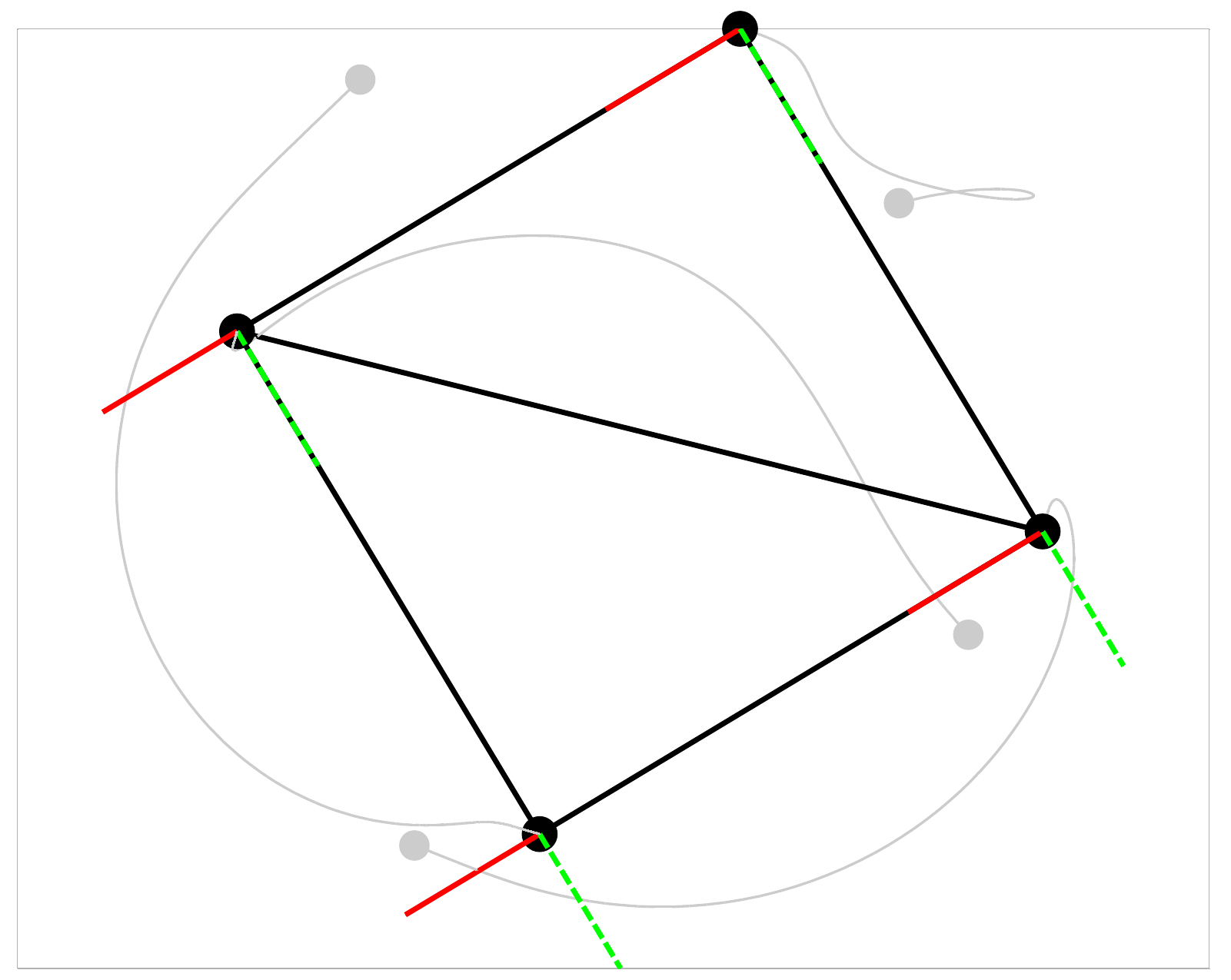}}
  %\vspace{-5pt}
  \caption{The case without a global reference frame in $\R^2$ with $n=4$, $m=5$.}
  \label{fig_sim_noGlobal_2DSquare}
    \vspace{-10pt}
\end{figure}
\begin{figure}[t]
  \centering
  \subfloat[Initial formation]{\includegraphics[width=0.35\linewidth]{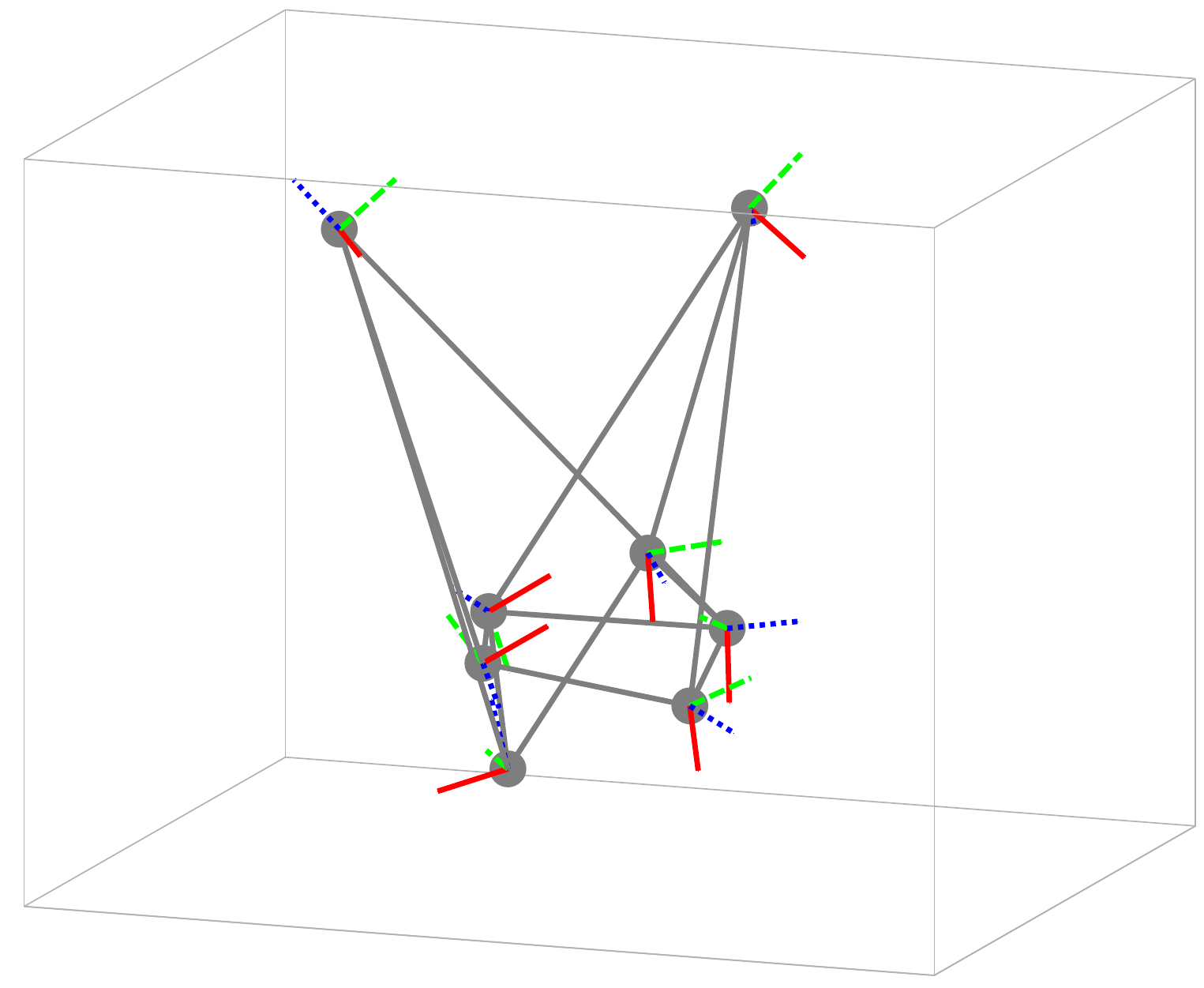}}
  \subfloat[Final formation]{\includegraphics[width=0.35\linewidth]{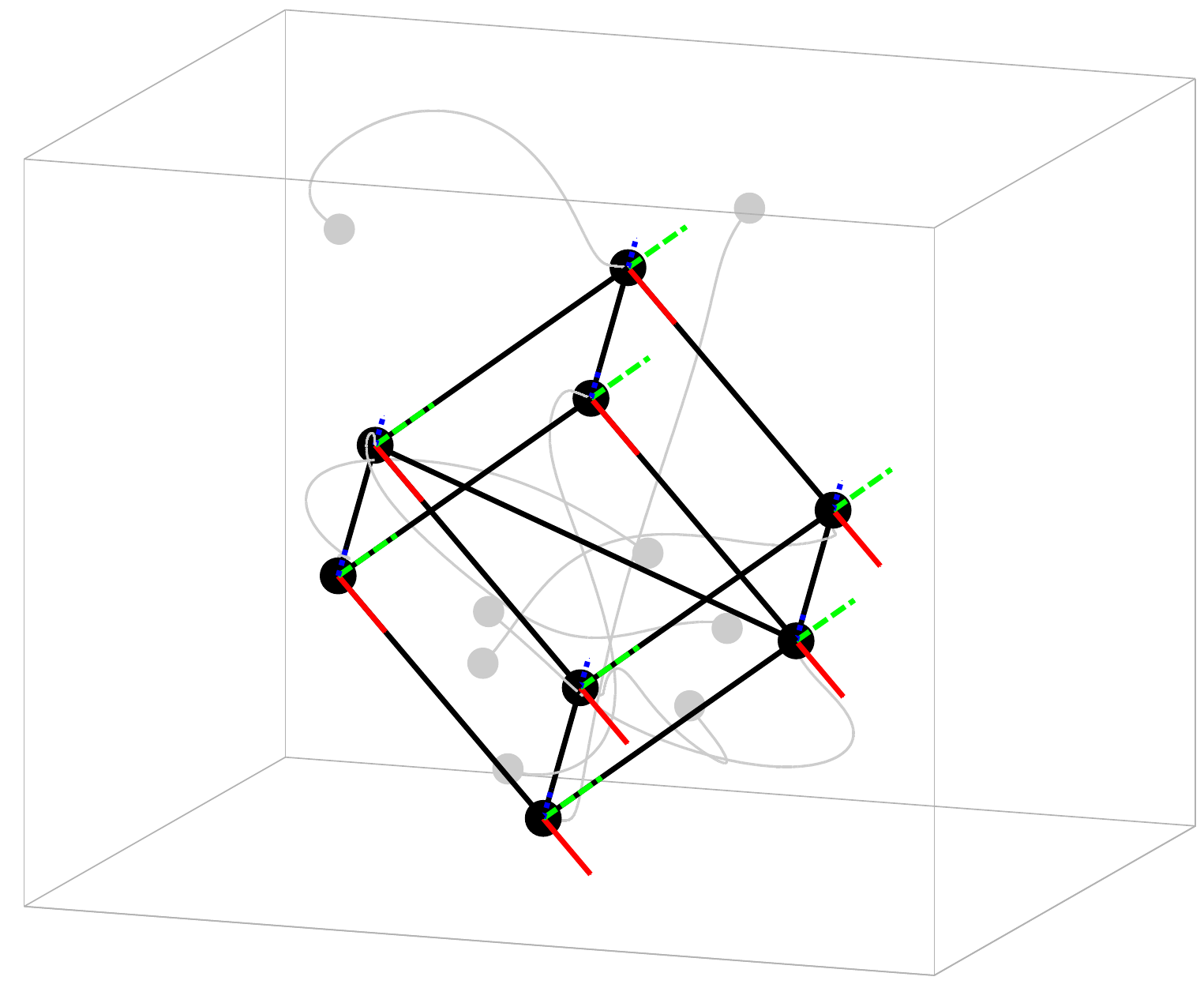}}
  %\vspace{-5pt}
  \caption{The case without a global reference frame in $\R^3$ with $n=8$, $m=13$.}
  \label{fig_sim_noGlobal_3DCube}
    \vspace{-10pt}
\end{figure}

\appendix

\subsection{Proof of Theorem~\ref{theorem_2DIPRandIREquivalence} and Corollary~\ref{corollary_infinitesimalBearngAndDistanceMotion}}\label{appendix_distanceRigidity}

In order to prove Theorem~\ref{theorem_2DIPRandIREquivalence} and Corollary~\ref{corollary_infinitesimalBearngAndDistanceMotion}, we need first introduce some concepts and results in the distance rigidity theory \cite{Hendrickson1992SIAM,Connelly2005Generic}.
Define the \emph{distance function} for a framework $\G(p)$ as
\begin{align}\label{eq_distanceEdgeFunction}
    F_D(p)\triangleq\frac{1}{2}
    \left[
      \begin{array}{ccc}
        \|e_1\|^2 & \cdots & \|e_m\|^2 \\
      \end{array}
    \right]^\T\in\R^{m}.
\end{align}
Each entry of $F_D(p)$ corresponds to the length of an edge of the framework.
The \emph{distance rigidity matrix} is defined as the Jacobian of the distance function,
\begin{align*}%\label{eq_DistanceRigidityMatrixDefinition}
    R_D(p) \triangleq \frac{\partial F_D(p)}{\partial p}\in\R^{m\times dn}.
\end{align*}
Let $\delta p$ be a variation of $p$.
If $R_D(p) \delta{p}=0$, then $\delta{p}$ is called an \emph{infinitesimal distance motion} of $\G(p)$.
A framework is \emph{infinitesimally distance rigid} if the infinitesimal motion only corresponds to rigid-body rotations and translations.

\begin{lemma}[\cite{Hendrickson1992SIAM}]\label{lemma_diatanceInfiRigid_NSCondition}
    A framework $\G(p)$ in $\R^d$ is infinitesimally distance rigid if and only if
    \begin{align*}
        \rank(R_D(p))=\left\{
                              \begin{array}{ll}
                                dn-d(d+1)/2 & \text{if $n\ge d$}, \\
                                n(n-1)/2 & \text{if $n<d$}. \\
                              \end{array}
                            \right.
    \end{align*}
\end{lemma}

In the case of $n\ge d$, the framework $\G(p)$ is infinitesimally distance rigid in $\R^2$ if and only if $\rank(R_D(p))=2n-3$, and in $\R^3$ if and only if $\rank(R_D(p))=3n-6$.

To prove Theorem~\ref{theorem_2DIPRandIREquivalence}, we first prove the following result which indicates that the bearing rigidity matrix always has the same rank as the distance rigidity matrix for any framework in $\R^2$.

\begin{proposition}\label{proposition_R2RBRDSameRank}
    For any framework $\G(p)$ in $\R^2$,  $\rank(R(p))=\rank(R_D(p))$.
\end{proposition}
\begin{proof}
        Consider an oriented graph and write the bearings of the framework as $\{g_k\}_{k=1}^m$.
    Let $Q_{\pi/2}$ be a $2\times2$ rotation matrix that rotates any vector $\pi/2$.
    Denote $g_k^\perp\triangleq Q_{\pi/2}g_k$. Then, $g_k^\perp\perp g_k$ and $\|g_k^\perp\|=\|g_k\|=1$.
    Since $P_{g_k}=g_k^\perp(g_k^\perp)^\T$, the bearing rigidity matrix can be rewritten as
    \begin{align}\label{eq_Rp_expression}
        R(p)
        =\dia{\frac{P_{g_k}}{\|e_k\|}}\bar{H}
        =\dia{\frac{g_{k}^\perp}{\|e_k\|}}\dia{(g_{k}^\perp)^\T}\bar{H}.
    \end{align}
    The matrix $\dia{(g_{k}^\perp)^\T}\bar{H}$ can be further written as
    \begin{align}\label{eq_Rp_expression1}
        &\dia{(g_{k}^\perp)^\T}\bar{H}=\dia{g_{k}^\T Q_{\pi/2}^\T}\bar{H}\nonumber\\
        &=\dia{g_{k}^\T}(I_m\otimes Q_{\pi/2}^\T)(H\otimes I_2)=\dia{g_{k}^\T}(H\otimes Q_{\pi/2}^\T) \nonumber\\ &=\dia{g_{k}^\T}\bar{H}(I_n\otimes Q_{\pi/2}^\T).
    \end{align}
    Note the distance rigidity matrix can be expressed as $R_D(p)=\dia{e_k^\T }\bar{H}$ (this expression can be obtained by calculating the Jacobian of the distance function \eqref{eq_distanceEdgeFunction}).
    Substituting $\dia{g_{k}^\T}\bar{H}= \dia{{1}/{\|e_k\|}}R_D(p)$ and \eqref{eq_Rp_expression1} into \eqref{eq_Rp_expression} yields
    \begin{align}\label{eq_RBAndRDRelation}
    R(p)=\dia{\frac{g_{k}^\perp}{\|e_k\|^2}}R_D(p)\left(I_n\otimes Q_{\pi/2}^\T\right).
    \end{align}
    Since $\dia{{g_{k}^\perp}/{\|e_k\|^2}}$ has full column rank and $I_n\otimes Q_{\pi/2}^\T$ is invertible, we have $\rank(R(p))=\rank(R_D(p))$.
\end{proof}

\begin{proof}[Proof of Theorem~\ref{theorem_2DIPRandIREquivalence}]
    By Theorem~\ref{theorem_conditionInfiParaRigid}, a framework $\G(p)$ in $\R^2$ is \emph{infinitesimally bearing rigid} if and only if $\rank(R(p))=2n-3$.
    By Lemma~\ref{lemma_diatanceInfiRigid_NSCondition}, a framework is \emph{infinitesimally distance rigid} if and only if $\rank(R_D(p))=2n-3$.
    Since $\rank(R(p))=\rank(R_D(p))$ as proved in Proposition~\ref{proposition_R2RBRDSameRank}, we know $\rank(R(p))=2n-3$ if and only if $\rank(R_D(p))=2n-3$, which concludes the theorem.
\end{proof}

\begin{proof}[Proof of Corollary~\ref{corollary_infinitesimalBearngAndDistanceMotion}]
It immediately follows from from \eqref{eq_RBAndRDRelation} that $R(p)\delta p=0$ if and only if $R_D(p)\delta p^\perp=0$.
\end{proof}

{\small
\section*{Acknowledgements}
The work presented here has been supported by the Israel Science Foundation (grant No. 1490/13).

%===========================================================
\bibliography{myOwnPub,zsyReferenceAll} % if no referece is cited before, there will be an error
\bibliographystyle{ieeetr}
}
\begin{IEEEbiography}[{\includegraphics[width=1in,height=1.25in,clip,keepaspectratio]{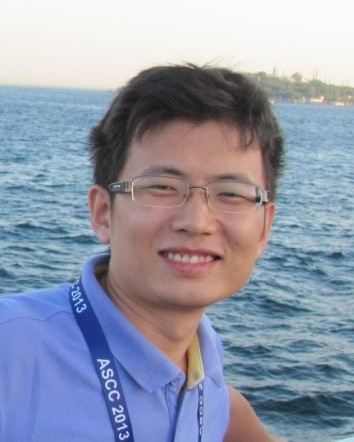}}]{Shiyu Zhao}
is a postdoctoral research fellow in the Faculty of Aerospace Engineering at the
Technion - Israel Institute of Technology.  He received his B.Eng (06) and
M.Eng (09) degrees from Beijing University of Aeronautics and Astronautics.
He got his Ph.D. in Electrical Engineering from National University of Singapore in 2014.
His research interests include distributed control and estimation of networked dynamical systems and its application to intelligent and robotic systems.
\end{IEEEbiography}

\begin{IEEEbiography}[{\includegraphics[width=1in,height=1.25in,clip,keepaspectratio]{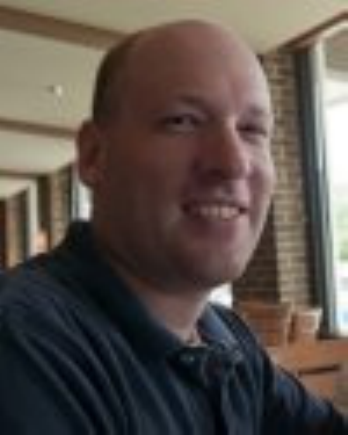}}]{Daniel Zelazo}
is an Assistant Professor of Aerospace Engineering at the
Technion - Israel Institute of Technology.  He received his BSc. (99) and
M.Eng (01) degrees in Electrical Engineering from the Massachusetts
Institute of Technology. In 2009, he completed his
Ph.D. from the University of Washington in Aeronautics
and Astronautics. From 2010-2012 he served as a post-doctoral research
associate and lecturer at the Institute for Systems Theory \& Automatic
Control in the University of Stuttgart. His research interests include
topics related to multi-agent systems, optimization, and graph
theory.
\end{IEEEbiography}

\end{document}

%% file: before_document.tex
\usepackage{graphicx}
\usepackage{subfig}
\usepackage{amsmath}
\usepackage{amsthm}
\usepackage{amssymb}
\usepackage[font=footnotesize]{caption} % set the captain font size to 8 (i.e. footnotesize)
\usepackage{subfig} % uses subfloats within a single float MUST after the package {caption}!!
\usepackage[noadjust]{cite} % alphabetize grouped citations automatically [3, 1, 2] to [1, 2, 3] % [noadjust] is to make ( [1]) to ([1])
\usepackage{color}
\usepackage{algorithm} % options: boxed [section]
\usepackage{algpseudocode} % for algorithm
\usepackage{enumerate}
\usepackage[hidelinks,colorlinks=false]{hyperref}
\usepackage{setspace}
\usepackage{tikz}
\usetikzlibrary{calc} % for calculation functions in Tikz let, in commands in Tikz
\usetikzlibrary{shapes} % for block diagram
\usetikzlibrary{chains}
\usetikzlibrary{fit}
\usetikzlibrary{arrows}
\usetikzlibrary{decorations.text} % text along path
\usetikzlibrary{decorations.markings}

\newtheorem{theorem}{Theorem}
\newtheorem{lemma}{Lemma}
\newtheorem{assumption}{Assumption}
\newtheorem{remark}{Remark}
\newtheorem{proposition}{Proposition}
\newtheorem{corollary}{Corollary}

\newtheorem{definition}{Definition}
\newtheorem{problem}{Problem}

\newcommand{\Null}{\mathrm{Null}}
\newcommand{\Range}{\mathrm{Range}}
\newcommand{\one}{\mathbf{1}}
\newcommand{\rank}{\mathrm{rank}}
\newcommand{\myspan}{\mathrm{span}}
\newcommand{\mydiag}{\mathrm{diag}}

\newcommand{\blkdiag}{\mathrm{blkdiag}}

\newcommand{\T}{\mathrm{T}}

\newcommand{\R}{\mathbb{R}}
\newcommand{\G}{\mathcal{G}}
\newcommand{\E}{\mathcal{E}}
\newcommand{\V}{\mathcal{V}}
\newcommand{\N}{\mathcal{N}}

\renewcommand{\S}{\mathcal{S}}

\newcommand{\sk}[1]{\left[#1\right]_\times} % skew symmetric operator
\newcommand{\dia}[1]{\mathrm{diag}\left(#1\right)} % block diagnal matrix

\graphicspath{{figures/}}

%% file: figure_tikz_equivNotCongru.tex
\def\mycolor{black}%{black!70!blue}
\subfloat[]{
\begin{tikzpicture}[scale=\myscale]
            \def\length{3}
            \coordinate (x1) at (0,0);
            \coordinate (x2) at (\length,0);
            \coordinate (x3) at (\length,-\length);
            \coordinate (x4) at (0,-\length);
            % draw lines
            \draw [] (x1)--(x2)--(x3)--(x4)--cycle;
            % draw circles %%%%%%%%%%%%%%%%%%%%%%%%%%%%%%%%%%%%%%%%%%%%%%%%%%%%%%%%%%%
            \def\radius{8pt}
            \draw [fill=white](x1) circle [radius=\radius];
            \draw (x1) node[] {\scriptsize{\color{\mycolor}$1$}};
            \draw [fill=white](x2) circle [radius=\radius];
            \draw (x2) node[] {\scriptsize{\color{\mycolor}$2$}};
            \draw [fill=white](x3) circle [radius=\radius];
            \draw (x3) node[] {\scriptsize{\color{\mycolor}$3$}};
            \draw [fill=white](x4) circle [radius=\radius];
            \draw (x4) node[] {\scriptsize{\color{\mycolor}$4$}};
\end{tikzpicture}
}
\qquad
\subfloat[]{
\begin{tikzpicture}[scale=\myscale]
            \def\length{3}
            \coordinate (x1) at (0,0);
            \coordinate (x2) at (\length/2*3,0);
            \coordinate (x3) at (\length/2*3,-\length);
            \coordinate (x4) at (0,-\length);
            % draw lines
            \draw [] (x1)--(x2)--(x3)--(x4)--cycle;
            % draw circles %%%%%%%%%%%%%%%%%%%%%%%%%%%%%%%%%%%%%%%%%%%%%%%%%%%%%%%%%%%
            \def\radius{8pt}
            \draw [fill=white](x1) circle [radius=\radius];
            \draw (x1) node[] {\scriptsize{\color{\mycolor}$1$}};
            \draw [fill=white](x2) circle [radius=\radius];
            \draw (x2) node[] {\scriptsize{\color{\mycolor}$2$}};
            \draw [fill=white](x3) circle [radius=\radius];
            \draw (x3) node[] {\scriptsize{\color{\mycolor}$3$}};
            \draw [fill=white](x4) circle [radius=\radius];
            \draw (x4) node[] {\scriptsize{\color{\mycolor}$4$}};
\end{tikzpicture}
}

%% file: figure_tikz_relationBetweenTheRigidities.tex
\begin{tikzpicture}[scale=\myscale]
\centering
\def\w{3}
\def\h{3}
\def\offset{2.5}
\def\hOffset{1.2}
\coordinate (x1) at (0,0);
\coordinate (x1_left) at (-\offset,0);
\coordinate (x1_right) at (\offset,0);
\coordinate (x2) at (-\w,-\h);
\coordinate (x2_right) at (-\w,-\h);
\coordinate (x2_up) at (-\w-\offset,-\h+\hOffset);
\coordinate (x3) at (\w,-\h);
\coordinate (x3_left) at (\w,-\h);
\coordinate (x3_up) at (\w+\offset,-\h+\hOffset);
% draw lines
\draw[implies-implies,double equal sign distance] (x2_right) -- (x3_left);
\draw[-implies,double equal sign distance] (x1_left) -- (x2_up);
\draw[-implies,double equal sign distance] (x1_right) -- (x3_up);
\draw[] (x1)node[above,align=center]{infinitesimal \\ bearing rigidity};
\draw[] (x2)node[left]{bearing rigidity};
\draw[] (x3)node[right,align=center]{global\\bearing rigidity};
\end{tikzpicture}

%\begin{tikzpicture}[scale=\myscale]
%\centering
%\def\w{7}
%\def\h{1.6}
%\def\wOffset{5}
%\def\hOffset{0.8}
%\coordinate (x1) at (0,0);
%\coordinate (x1_up) at (0,\hOffset);
%\coordinate (x1_down) at (0,-\hOffset);
%\coordinate (x2) at (\w,\h);
%\coordinate (x2_left) at (\w-\wOffset,\h);
%\coordinate (x2_down) at (\w,\h-\hOffset);
%\coordinate (x3) at (\w,-\h);
%\coordinate (x3_left) at (\w-\wOffset,-\h);
%\coordinate (x3_up) at (\w,-\h+\hOffset);
%% draw lines
%\draw[] (x1)node[left,align=center]{infinitesimal \\ bearing rigidity};
%\draw[] (x2)node[]{bearing rigidity};
%\draw[] (x3)node[align=center]{global bearing rigidity};
%\draw[implies-implies,double equal sign distance] (x2_down) -- (x3_up);
%\draw[-implies,double equal sign distance] (x1_up) -- (x2_left);
%\draw[-implies,double equal sign distance] (x1_down) -- (x3_left);
%\end{tikzpicture}

%% file: fig_tikz_Example_nonIBR.tex
\def\arrowL{0.6}
\def\arrowWidth{semithick}
\subfloat[]{
\begin{tikzpicture}[scale=\myscale]
\def\length{3}
\def\offset{0.15}
\coordinate (x1) at (0,0);
\coordinate (x1Above) at (0,\offset);
\coordinate (x1Below) at (0,-\offset);
\coordinate (x2) at (\length/2,\offset);
\coordinate (x3) at (\length,0);
\coordinate (x3Above) at (\length,\offset);
\coordinate (x3Below) at (\length,-\offset);
%draw lines
\draw [] (x1Above)--(x2)--(x3Above);
\draw [] (x1Below)--(x3Below);
\draw [\arrowWidth,->,draw=red] (x2)--(\length/2+\arrowL,\offset);
\draw [densely dotted,\arrowWidth,->,draw=blue] (x2)--(\length/2,\offset+\arrowL);
% draw circles
\def\radius{5pt}
\draw [fill=white](x1) circle [radius=\radius];
\draw [fill=white](x2) circle [radius=\radius];
\draw [fill=white](x3) circle [radius=\radius];
\end{tikzpicture}}
%%%%%%%%%%%%%%%%%%%%%%%%%%%%%%%%%%%%%%%%%%%%%%%%%%%%%%%%%%%
\subfloat[]{
\begin{tikzpicture}[scale=\myscale]
            \def\length{3}
            \coordinate (x1) at (0,0);
            \coordinate (x2) at (\length,0);
            \coordinate (x3) at (\length,-\length);
            \coordinate (x4) at (0,-\length);
            % draw lines
            \draw [] (x1)--(x2)--(x3)--(x4)--cycle;
            \draw [\arrowWidth,->,draw=red] (x1)--(0,\arrowL);
            \draw [densely dotted,\arrowWidth,->,draw=blue] (x1)--(\arrowL,0);
            \draw [\arrowWidth,->,draw=red] (x2)--(\length,\arrowL);
            \draw [densely dotted,\arrowWidth,->,draw=blue] (x2)--(\length+\arrowL,0);
            % draw circles %%%%%%%%%%%%%%%%%%%%%%%%%%%%%%%%%%%%%%%%%%%%%%%%%%%%%%%%%%%
            \def\radius{5pt}
            \draw [fill=white](x1) circle [radius=\radius];
            \draw [fill=white](x2) circle [radius=\radius];
            \draw [fill=white](x3) circle [radius=\radius];
            \draw [fill=white](x4) circle [radius=\radius];
\end{tikzpicture}
}
%%%%%%%%%%%%%%%%%%%%%%%%%%%%%%%%%%%%%%%%%%%%%%%%%%%%%%%%%%%
\subfloat[]{
\begin{tikzpicture}[scale=\myscale]
            % scale of the two triangles
            \def\scaleBig{2.5}; % big triangle
            \coordinate (x1) at (-0.866*\scaleBig,-0.5*\scaleBig);
            \coordinate (x2) at (0.866*\scaleBig,-0.5*\scaleBig);
            \coordinate (x3) at (0,\scaleBig);
            \def\scaleSmall{1.2}; % small triangle
            \coordinate (x4) at (-0.866*\scaleSmall,-0.5*\scaleSmall);
            \coordinate (x5) at (0.866*\scaleSmall,-0.5*\scaleSmall);
            \coordinate (x6) at (0,\scaleSmall);
            % draw lines
            \draw [] (x1)--(x4);
            \draw [] (x2)--(x5);
            \draw [] (x3)--(x6);
            \draw [] (x1)--(x2)--(x3)--cycle;
            \draw [] (x4)--(x5)--(x6)--cycle;
            \draw [\arrowWidth,->,draw=red] (x4)--(-0.866*\scaleSmall-0.866*\arrowL,-0.5*\scaleSmall-0.5*\arrowL);
            \draw [densely dotted,\arrowWidth,->,draw=blue] (x4)--(-0.866*\scaleSmall-0.5*\arrowL,-0.5*\scaleSmall+0.866*\arrowL);
            \draw [\arrowWidth,->,draw=red] (x5)--(0.866*\scaleSmall+0.866*\arrowL,-0.5*\scaleSmall-0.5*\arrowL);
            \draw [densely dotted,\arrowWidth,->,draw=blue] (x5)--(0.866*\scaleSmall-0.5*\arrowL,-0.5*\scaleSmall-0.866*\arrowL);
            \draw [\arrowWidth,->,draw=red] (x6)--(0,\scaleSmall+\arrowL);
            \draw [densely dotted,\arrowWidth,->,draw=blue] (x6)--(\arrowL,\scaleSmall);
            % draw circles %%%%%%%%%%%%%%%%%%%%%%%%%%%%%%%%%%%%%%%%%%%%%%%%%%%%%%%%%%%
            \def\radius{5pt}
            \draw [fill=white](x1) circle [radius=\radius];
            \draw [fill=white](x2) circle [radius=\radius];
            \draw [fill=white](x3) circle [radius=\radius];
            \draw [fill=white](x4) circle [radius=\radius];
            \draw [fill=white](x5) circle [radius=\radius];
            \draw [fill=white](x6) circle [radius=\radius];
\end{tikzpicture}}
%%%%%%%%%%%%%%%%%%%%%%%%%%%%%%%%%%%%%%%%%%%%%%%%%%%%%%
\subfloat[]{
\begin{tikzpicture}[scale=\myscale]
\def\length{3}
\def\high{3}
\coordinate (x1) at (0,0);
\coordinate (x2) at (\length,0);
\coordinate (x3) at (\length,\high);
\coordinate (x4) at (0,\high);
\coordinate (x5) at (\length/3,\high/2);
\coordinate (x6) at (\length/3*2,\high/2);
%draw lines
\draw [] (x1)--(x2)--(x3)--(x4)--cycle;
\draw [] (x1)--(x5)--(x4);
\draw [] (x2)--(x6)--(x3);
\draw [] (x5)--(x6);
\draw [\arrowWidth,->,draw=red] (x2)--(\length+\arrowL,0);
\draw [densely dotted,\arrowWidth,->,draw=blue] (x2)--(\length, \arrowL);
\draw [\arrowWidth,->,draw=red] (x3)--(\length+\arrowL,\high);
\draw [densely dotted,\arrowWidth,->,draw=blue] (x3)--(\length, \high+\arrowL);
\draw [\arrowWidth,->,draw=red] (x6)--(\length/3*2+\arrowL,\high/2);
\draw [densely dotted,\arrowWidth,->,draw=blue] (x6)--(\length/3*2,\high/2+\arrowL);
% draw circles %%%%%%%%%%%%%%%%%%%%%%%%%%%%%%%%%%%%%%%%%%%%%%%%%%%%%%%%%%%
\def\radius{5pt}
\draw [fill=white](x1) circle [radius=\radius];
\draw [fill=white](x2) circle [radius=\radius];
\draw [fill=white](x3) circle [radius=\radius];
\draw [fill=white](x4) circle [radius=\radius];
\draw [fill=white](x5) circle [radius=\radius];
\draw [fill=white](x6) circle [radius=\radius];
\end{tikzpicture}}

%% file: figure_tikz_IPRExamples.tex
\subfloat[]{
\begin{tikzpicture}[scale=\myscale]
\def\length{3.5}
            \coordinate (x1) at (0,0);
            \coordinate (x2) at (\length,0);
            \coordinate (x3) at (\length/2,3);
            % draw lines
            \draw [] (x1)--(x2)--(x3)--cycle;
            % draw circles %%%%%%%%%%%%%%%%%%%%%%%%%%%%%%%%%%%%%%%%%%%%%%%%%%%%%%%%%%%
            \def\radius{5pt}
            \draw [fill=white](x1) circle [radius=\radius];
            \draw [fill=white](x2) circle [radius=\radius];
            \draw [fill=white](x3) circle [radius=\radius];
\end{tikzpicture}}
%%%%%%%%%%%%%%%%%%%%%%%%%%%%%%%%%%%%%%%%%%%%
\subfloat[]{
\begin{tikzpicture}[scale=\myscale]
            \def\length{3}
%            \coordinate (xInvisible1) at (-\length/4,0); % for the subcaption
%            \coordinate (xInvisible2) at (5*\length/4,0);
            \coordinate (x1) at (0,0);
            \coordinate (x2) at (\length,0);
            \coordinate (x3) at (\length,-\length);
            \coordinate (x4) at (0,-\length);
            % draw lines
%            \draw[draw opacity=0](xInvisible1)--(xInvisible2); % transparent line for the long subcaption
            \draw [] (x1)--(x2)--(x3)--(x4)--cycle;
            \draw [] (x1)--(x3);
            % draw circles %%%%%%%%%%%%%%%%%%%%%%%%%%%%%%%%%%%%%%%%%%%%%%%%%%%%%%%%%%%
            \def\radius{5pt}
            \draw [fill=white](x1) circle [radius=\radius];
            \draw [fill=white](x2) circle [radius=\radius];
            \draw [fill=white](x3) circle [radius=\radius];
            \draw [fill=white](x4) circle [radius=\radius];
\end{tikzpicture}
}
\subfloat[]{
\begin{tikzpicture}[scale=\myscale]
            \def\length{2.6}
            \coordinate (x1) at (0,0);
            \coordinate (x2) at (\length,0);
            \coordinate (x3) at (\length,-\length);
            \coordinate (x4) at (0,-\length);
            \def\Xoffset{1}
            \def\Yoffset{1}
            \coordinate (x5) at (0+\Xoffset,0+\Yoffset);
            \coordinate (x6) at (\length+\Xoffset,0+\Yoffset);
            \coordinate (x7) at (\length+\Xoffset,-\length+\Yoffset);
            \coordinate (x8) at (0+\Xoffset,-\length+\Yoffset);
            % draw lines
            \draw [] (x1)--(x2)--(x3)--(x4)--cycle;
            \draw [] (x1)--(x5)--(x6)--(x2);
            \draw [] (x6)--(x7)--(x3);
            \draw [densely dotted] (x5)--(x8);
            \draw [densely dotted] (x7)--(x8);
            \draw [densely dotted] (x4)--(x8);
            \draw [densely dotted] (x1)--(x7);
            % draw circles %%%%%%%%%%%%%%%%%%%%%%%%%%%%%%%%%%%%%%%%%%%%%%%%%%%%%%%%%%%
            \def\radius{5pt}
            \draw [fill=white](x1) circle [radius=\radius];
            \draw [fill=white](x2) circle [radius=\radius];
            \draw [fill=white](x3) circle [radius=\radius];
            \draw [fill=white](x4) circle [radius=\radius];
            \draw [fill=white](x5) circle [radius=\radius];
            \draw [fill=white](x6) circle [radius=\radius];
            \draw [fill=white](x7) circle [radius=\radius];
            \draw [fill=white,densely dotted](x8) circle [radius=\radius];
%            \draw (x1) node[above=\radius/2] {1};
%            \draw (x2) node[above=\radius/2] {2};
%            \draw (x3) node[below=\radius/2] {3};
%            \draw (x4) node[below=\radius/2] {4};
\end{tikzpicture}
}
\subfloat[]{
\begin{tikzpicture}[scale=\myscale]
            \def\length{1}
            \coordinate (x1) at (-\length,-2/3*\length);
            \coordinate (x2) at (\length,-2/3*\length);
            \coordinate (x3) at (2*\length,0);
            \coordinate (x4) at (\length-0.05,2/3*\length-0.1);
            \coordinate (x5) at (-\length+0.05,2/3*\length-0.1);
            \coordinate (x6) at (-2*\length,0);
            \coordinate (x7) at (0,3*\length);
            % draw lines
            \draw [] (x6)--(x1)--(x2)--(x3);
            \draw [densely dotted] (x3)--(x4)--(x5)--(x6);
            \draw [] (x1)--(x7);
            \draw [] (x2)--(x7);
            \draw [] (x3)--(x7);
            \draw [densely dotted] (x4)--(x7);
            \draw [densely dotted] (x5)--(x7);
            \draw [] (x6)--(x7);
%            \draw [densely dotted] (x5)--(x8);
%            \draw [densely dotted] (x7)--(x8);
%            \draw [densely dotted] (x4)--(x8);
%            \draw [densely dotted] (x1)--(x7);
            % draw circles %%%%%%%%%%%%%%%%%%%%%%%%%%%%%%%%%%%%%%%%%%%%%%%%%%%%%%%%%%%
            \def\radius{5pt}
            \draw [fill=white](x1) circle [radius=\radius];
            \draw [fill=white](x2) circle [radius=\radius];
            \draw [fill=white](x3) circle [radius=\radius];
            \draw [fill=white,densely dotted](x4) circle [radius=\radius];
            \draw [fill=white,densely dotted](x5) circle [radius=\radius];
            \draw [fill=white](x6) circle [radius=\radius];
            \draw [fill=white](x7) circle [radius=\radius];
\end{tikzpicture}
}

%% file: figure_tikz_controlLawMeaning.tex
\begin{tikzpicture}[scale=\myscale]
            \coordinate (xi) at (0,0);
            \coordinate (xj) at (6,6);
            \def\unit{5.5cm}
            \coordinate (gij)                 at (\unit*0.7071,\unit*0.7071);
            \coordinate (gij_star)            at (\unit,0);
            \coordinate (gij_star_proj)       at (\unit*0.7071*0.7071,-\unit*0.7071*0.7071);
            \coordinate (gij_star_proj_minus) at (-\unit*0.7071*0.7071,\unit*0.7071*0.7071);
            %%%%%%%%%%%%%%%%%%%%%%%%%%%%%%%%%%%%%%%%%%%%%%%%%%%%%
            % draw lines
            \draw [densely dotted, thin] (xi)--(xj);
            \draw [densely dotted, thin] (-5,5)--(4,-4);
            \draw [densely dashed, thick] (gij_star)--(gij_star_proj);
            %%%%%%%%%%%%%%%%%%%%%%%%%%%%%%%%%%%%%%%%%%%%%%%%%%%%%
            % draw arrows
            \draw[->, >=latex, thick] (xi) -- (gij) node [left=4pt] {\scriptsize{$g_{ij}$}};
            \draw[->, >=latex, thick] (xi) -- (gij_star) node [right] {\scriptsize{$g_{ij}^*$}};
            %\draw[->, >=latex, thick] (gij_star) -- (gij) node [right=4pt] {$\varepsilon_{i}$};
            \draw[->, >=latex, thick] (xi) -- (gij_star_proj) node [left=4pt] {\scriptsize{$P_{g_{ij}}g^*_{ij}$}};
            \draw[->, >=latex, very thick] (xi) -- (gij_star_proj_minus) node [left] {\scriptsize{$-P_{g_{ij}}g^*_{ij}$}};
            %%%%%%%%%%%%%%%%%%%%%%%%%%%%%%%%%%%%%%%%%%%%%%%%%%%%%%%%%%%
            % draw circles
            \def\radius{8pt}
            \draw [fill=white](xi) circle [radius=\radius];
            \draw [fill=white](xj) circle [radius=\radius];
            \draw (xi) node[below left] {\scriptsize{$p_i$}};
            \draw (xj) node[left=\radius/2] {\scriptsize{$p_j$}};
            %%%%%%%%%%%%%%%%%%%%%%%%%%%%%%%%%%%%%%%%%%%%%%%%%%%%%%%%%%%
            % draw the right angle
            \def\dis{1cm}
            \coordinate (p1) at (0,\dis);
            \coordinate (p2) at (\dis*0.5,\dis*0.5);
            \coordinate (p3) at (-\dis*0.5,\dis*0.5);
            \draw [] (p1)--(p2);
            \draw [] (p1)--(p3);
            \coordinate (p4) at (\unit*0.7071*0.7071,-\unit*0.7071*0.7071+\dis);
            \coordinate (p5) at (\unit*0.7071*0.7071+\dis/2,-\unit*0.7071*0.7071+\dis/2);
            \coordinate (p6) at (\unit*0.7071*0.7071-\dis/2,-\unit*0.7071*0.7071+\dis/2);
            \draw [] (p4)--(p5);
            \draw [] (p4)--(p6);
\end{tikzpicture}

%% file: figure_tikz_undesiredEuqilibrium.tex
\begin{tikzpicture}[scale=\myscale]
\def\length{2}
\def\linetype{}
\coordinate (origin) at (0,0);
\coordinate (x1) at (-\length,-\length/2);
\coordinate (x2) at (0,\length);
\coordinate (x3) at (\length,-\length/2);
\coordinate (x1_ref) at (\length,\length/2);
\coordinate (x2_ref) at (0,-\length);
\coordinate (x3_ref) at (-\length,\length/2);
\def\radius{8pt}
\draw [fill=black](origin) circle [radius=1pt];
%draw triangle %%%%%%%%%%%%%%%%%%%%%%%%%%%%%%%%%%%%%%%%%%
\def\mycolor{black}%{black!70!blue}
\draw [thick, draw=\mycolor] (x1)--(x2)--(x3)--cycle;
\draw [thick,fill=white, draw=\mycolor](x1) circle [radius=\radius];
\draw (x1) node[] {\scriptsize{\color{\mycolor}$1$}};
\draw [thick,fill=white, draw=\mycolor](x2) circle [radius=\radius];
\draw (x2) node[] {\scriptsize{\color{\mycolor}$2$}};
\draw [thick,fill=white, draw=\mycolor](x3) circle [radius=\radius];
\draw (x3) node[] {\scriptsize{\color{\mycolor}$3$}};
%draw reflection %%%%%%%%%%%%%%%%%%%%%%%%%%%%%%%%%%%%%%%%%%%%%%%%%%%%%%%%%%%
\def\mycolor2{black}%{black!70!green}
\draw [semithick,draw=\mycolor2,densely dotted] (x1_ref)--(x2_ref)--(x3_ref)--cycle;
\draw [semithick,fill=white,draw=\mycolor2,densely dotted](x1_ref) circle [radius=\radius];
\draw (x1_ref) node[] {\color{\mycolor2}\scriptsize{$1$}};
\draw [semithick,fill=white,draw=\mycolor2,densely dotted](x2_ref) circle [radius=\radius];
\draw (x2_ref) node[] {\color{\mycolor2}\scriptsize{$2$}};
\draw [semithick,fill=white,draw=\mycolor2,densely dotted](x3_ref) circle [radius=\radius];
\draw (x3_ref) node[] {\color{\mycolor2}\scriptsize{$3$}};
\end{tikzpicture}